\documentclass[journal,twoside,web]{ieeecolor}

\usepackage{generic}

\usepackage{amsmath} 
\usepackage{amsthm}
\usepackage{cases}

\usepackage{amsfonts}
\usepackage{amssymb}  
\usepackage{mathtools} 
\usepackage{mathrsfs}
\usepackage{balance}
\usepackage{nicefrac}
\usepackage{cite}
\usepackage{hyperref}
\hypersetup{
    colorlinks=true,
    linkcolor=blue,
    filecolor=magenta,      
    urlcolor=cyan,
}
\usepackage{graphicx}
\usepackage{textcomp}
\usepackage{color}

\let\labelindent\relax
\usepackage{enumitem}

\DeclareMathAlphabet{\mathpzc}{OT1}{pzc}{m}{it}

\newcommand{\R}{\mathbb{R}}

\theoremstyle{definition}

\newtheorem{definition}{Definition}
\newtheorem{problem}{Problem}

\newtheorem{lemma}{Lemma}
\newtheorem{theorem}{Theorem}
\newtheorem{remark}{Remark}
\newtheorem{assumption}{Assumption}

\usepackage{cleveref}

\usepackage[ruled,linesnumbered]{algorithm2e}

\newcommand{\norm}[1]{\left\lVert #1\right\rVert}

\DeclareMathOperator*{\argmin}{arg\,min}

\newcommand{\doi}[1]{\href{http://dx.doi.org/#1}{\normalsize{\textsc{doi:}}~\nolinkurl{#1}}}
\newcommand{\arxiv}[1]{\href{http://arxiv.org/abs/#1}{\normalsize{\textsc{arxiv:}}~\nolinkurl{#1}}}
\usepackage{color}
\usepackage{comment}
\newcommand{\HRule}{\noindent\rule{\linewidth}{0.1mm}\newline}

\renewcommand{\phi}{\varphi}

\renewcommand{\implies}{\;\Rightarrow\;}

\newcommand{\gpdomain}{\mathcal{X}}
\newcommand{\gpaugdomain}{\bar{\mathcal{X}}}
\newcommand{\gpvar}{x}
\newcommand{\augu}{y}
\newcommand{\augU}{\mathcal{Y}}

\newcommand{\gpmuB}{\mu_{B}}
\newcommand{\gpsigmaB}{\sigma_{B}}
\newcommand{\gpGramB}{\Sigma_{B}}
\newcommand{\gpGramUB}{\Sigma_{L_g B}}
\newcommand{\gpGramfB}{\Sigma_{L_f B}}

\newcommand{\socA}{\gpGramUB^{1/2}(x|\mathbb{D}_N)}
\newcommand{\socb}{\gpGramfB^{1/2}(x|\mathbb{D}_N)}
\newcommand{\socAA}{\gpGramUB(x|\mathbb{D}_N)}
\newcommand{\socbb}{\gpGramfB(x|\mathbb{D}_N)}
\newcommand{\socc}{\widehat{L_g B}(x|\mathbb{D}_N)}
\newcommand{\wfeas}{\widehat{L_g B}}
\newcommand{\socd}{\widehat{L_f B}(x|\mathbb{D}_N) + \gamma(B(x))}
\newcommand{\feasmat}{\mathcal{F}(x|\mathbb{D}_N)}

\newcommand{\eigval}{\lambda_\dagger}
\newcommand{\eigvec}{e_\dagger}
\newcommand{\usafe}{u_{\text{safe}}}
\newcommand{\thres}{\varepsilon}
\newcommand{\ualgo}{\bar{u}}

\newcommand{\revision}[1]{{\color{black} #1}}
\newcommand{\textblack}[1]{{\color{black} #1}}

\def\BibTeX{{\rm B\kern-.05em{\sc i\kern-.025em b}\kern-.08em
    T\kern-.1667em\lower.7ex\hbox{E}\kern-.125emX}}
\begin{document}

\title{\LARGE \bf
Recursively Feasible Probabilistic Safe Online Learning \\with Control Barrier Functions

}

\author{Fernando Castañeda*, Jason J. Choi*, Wonsuhk Jung, Bike Zhang, Claire J. Tomlin, and Koushil Sreenath
\thanks{* The first two authors contributed equally to the work.}
\thanks{F. Castañeda, J. J. Choi, B. Zhang, C. J. Tomlin, and K. Sreenath are with the University of California, Berkeley, CA, 94720, USA. \tt \{fcastaneda, jason.choi, bikezhang, tomlin, koushils\}@berkeley.edu }
\thanks{W. Jung is with Georgia Institute of Technology, GA, 30332, USA. \tt wonsuhk.jung@gatech.edu}
\thanks{This work was partially supported through National Science Foundation Grant CMMI-1931853 and CMMI-1944722. The work of Fernando Casta\~neda received the support of a fellowship from Fundaci\'on Rafael del Pino, Spain.}
}

\maketitle

\begin{abstract}
Learning-based control has recently shown great efficacy in performing complex tasks for various applications. However, to deploy it in real systems, it is of vital importance to guarantee the system will stay safe. Control Barrier Functions (CBFs) offer mathematical tools for designing safety-preserving controllers for systems with known dynamics. In this article, we first introduce a model-uncertainty-aware reformulation of CBF-based safety-critical controllers using Gaussian Process (GP) regression to close the gap between an approximate mathematical model and the real system, which results in a second-order cone program (SOCP)-based control design. We then present the pointwise feasibility conditions of the resulting safety controller, highlighting the level of richness that the available system information must meet to ensure safety. We use these conditions to devise an event-triggered online data collection strategy that ensures the recursive feasibility of the learned safety controller. Our method works by constantly reasoning about whether the current information is sufficient to ensure safety or if new measurements under active safe exploration are required to reduce the uncertainty. As a result, our proposed framework can guarantee the forward invariance of the safe set defined by the CBF with high probability, even if it contains a priori unexplored regions. We validate the proposed framework in two numerical simulation experiments.
\end{abstract}

\section{Introduction}
\label{sec:01introduction}

\subsection{Motivation \& Key Idea}
In many real-world control systems, such as aircrafts, industrial robots, or autonomous vehicles, in order to prevent system failure and catastrophic events, it is crucial to ensure that the system always stays within a set of safe states. Mathematical models of the system's dynamics can often be useful to design controllers that can enforce such safety constraints. However, since designing accurate models for complex real systems is not easy, imperfect models are typically used in practice, and the guarantees of the designed controllers can be lost when these model imperfections are not addressed appropriately. 

On the other hand, data-driven control has emerged as a new paradigm for solving complex control tasks. Nevertheless, in the absence of interpretable model-based knowledge, such methods usually fall short of theoretical guarantees. 
Moreover, data-driven approaches typically require collecting enough real-world data to accurately characterize the system. \revision{This presents a dilemma for safety-critical systems: we need to deploy the system to collect data, but we cannot deploy it safely without already having sufficient data.}

In this paper, we present an approach to address this dilemma and guarantee the safety of systems with uncertain dynamics. Our methodology lies at the intersection of model-based and data-driven control.
An imperfect dynamics model is complemented by the information gathered from data collected safely online on the real system, which allows us to ensure safety without a perfect model or offline data. 

\revision{The crucial component of our approach is determining whether the combination of prior model knowledge and collected data can maintain low uncertainty in a safe control direction or if new information is needed. To intuitively illustrate the working principle, consider an adaptive cruise control problem where the vehicle on a slippery road must maintain a safe distance from the car in front. We want to ensure that at all times the ego car has enough information about its dynamics reacting to a safe control action (e.g., braking). If the braking effect is well understood, the vehicle is allowed to follow the driver's commands. However, if the braking uncertainty reaches a critical level---after exceeding this, maintaining distance to the front car is no longer ensured despite braking---our method commands the ego vehicle to brake, allowing it to measure the braking effect and improve confidence. In this way, the vehicle is always prepared to prevent a collision with the car ahead, as the braking effect is sufficiently characterized at all times.}

\subsection{Related Work}

In the model-based control literature, various approaches exist for the design of controllers satisfying safety constraints, including Control Barrier Functions (CBFs, \revision{\cite{prajna2004safety, ames2017cbf}}), Hamilton-Jacobi Reachability \cite{bansal2017hamilton}, and Model Predictive Control \cite{wabersich2021predictive} to name a few. In this article, we focus on the CBF-based implementations of safe controllers for nonlinear systems. The main advantages of using CBFs for safety-critical control are twofold. First, the zero-superlevel set of a CBF, which is control invariant, explicitly verifies the state domain where safety is guaranteed. Second, while guaranteeing set invariance requires long horizon reasoning, CBFs condense this problem into a simple single time-step condition that should be satisfied at each time, similar to Lyapunov-based methods for stability.
This single time-step constraint on the control input derived from a CBF guarantees that the system does not exit the boundary of the zero-superlevel set and thus, remains safe.

Importantly, such safety constraints based on CBFs depend on the dynamics of the system. This means that when the dynamics are uncertain, the usage of an incorrect model in the CBF-based controller might lead to violation of safety. This issue can be tackled with an online adaptation of the control law \cite{nguyen2015l1adaptive, taylor2020adaptive, lopez2020robust, black2022adaptation}; however, these approaches usually assume the uncertainties have some restrictive structure.
Robust control approaches instead consider the worst-case effects of model uncertainty \cite{kolathaya2018input, nguyen2021robust, krstic2021inverse, choi2021robust,  alan2022control}. However, an inaccurate characterization of the disturbance bounds can lead to the violation of safety when the bounds are too optimistic, or can lead to impractical conservative behaviors when they are unnecessarily large.

These limitations allude to the core motivation of using modern data-driven techniques to address the effects of model mismatch: learn and adapt to the uncertainties with minimal structural assumptions, and learn the correct magnitude of the robustness bounds. \revision{For this purpose, recent research has extensively employed neural networks to learn the mismatch terms \cite{taylor2019clflearning, taylor2020cbf, westenbroek2020learning, choi2020reinforcement}.} While these approaches have been shown to be practical and effective, addressing the potential errors in neural network predictions remains challenging. \revision{Other works, like this work, use non-parametric regression methods, most notably Gaussian Process (GP) regression, that provide a probabilistic guarantee of the prediction quality under mild assumptions \cite{berkenkamp2016lyapunov, berkenkamp2017saferl, fisac2018general, umlauft2018clf, fan2019balsa, cheng2020safe, cohen2021safe, taylor2020towards, greeff2021learning, dhiman2021control, brunke2022barrier}. However, collecting the data for the regression would require exciting the system in many control directions, which might compromise safety. Guaranteeing safety for uncertain systems by using only data that can be safely collected remains an open problem.}

\revision{Finally, it is also worth mentioning existing works on different but related problems. While learning-based approaches have emerged as an effective paradigm for designing CBFs for uncertain systems \cite{dawson2022safe, qin2022sablas, jagtap2020control, lindemann2021learning, jin2020neural, nejati2023data}, it should be noted that the implementation of the CBF-based safety constraint still requires the knowledge of the system dynamics \cite{dawson2022safe}, unless the learned filter is completely model-free \cite{qin2022sablas, lavanakul2024safety}. The problem we focus on is the realization of safety filters for uncertain systems, not the design of CBFs; thus, our work is complementary to the line of work on data-driven CBF design (see Remark \ref{remark:learning-based-cbf} for more details). CBF-based methods have also been proposed to address stochastic systems \cite{clark2021control, so2023almost, cosner2024bounding} or systems with distributional uncertainty \cite{mestres2023feasibility}. Although these lines of work involve probabilistic analysis similar to ours, we focus on a problem set up under deterministic dynamics.
}

\subsection{Contributions}

Our work uses GP regression to learn the effects of an uncertain dynamics model on the CBF-based safety constraint. Then, a second-order cone program (SOCP)-based controller, namely the GP-CBF-SOCP, is proposed that gives a probabilistic safety guarantee when the optimization is feasible. \revision{This SOCP controller design was first introduced in the conference version of this work in \cite{castaneda2021pointwise}, which extended the work in \cite{GPCLFSOCP} for Lyapunov controllers.}

In \cite{castaneda2021pointwise}, we established necessary and sufficient conditions for the \revision{pointwise} feasibility of the SOCP controller \revision{at a} \revision{state. This pointwise feasibility condition implies that the data quality is \textit{locally} sufficient to ensure safety at that state. However, we only achieved a partial linkage between data quality and safety due to the lack of a \textit{recursive} feasibility guarantee. That is, with a fixed dataset, the SOCP controller could navigate to state-space regions with insufficient data, resulting in a loss of feasibility and potential safety violations. Frameworks in \cite{dhiman2021control, brunke2022barrier} face this same issue.}

\revision{The main component we introduce in this journal version is an event-triggered online data collection mechanism that ensures the recursive feasibility of the SOCP controller. By achieving this, we address the missing linkage between data quality, SOCP controller feasibility, and safety: we now use the feasibility analysis to guide data collection. Our approach tightly connects the three entities---evaluating the feasibility condition to judge if and how to improve the data, using the data for predictions with the GP model, and finally ensuring safety recursively---thereby establishing a complete online learning-based safety framework for uncertain systems.}

\revision{Specifically,} our online data collection algorithm ensures at all times the availability of a control input direction that renders the system safe with high probability. If the available prior knowledge from the model and past data is sufficient to characterize such a \revision{safe} direction, our proposed method simply acts as a safety filter applied to a performance-driven control law. However, whenever the uncertainty in the safe control direction reaches a critical level, our algorithm takes a safe exploration action that improves the knowledge of the system's response to such control inputs. Unlike the strategy in \cite{umlauft2019feedback}, which aims to improve overall accuracy of the GP prediction for a feedback linearization-based controller, \revision{our approach focuses on exciting safe control directions to collect new data that reduces uncertainty specifically in such directions.}

Moreover, we prove local Lipschitz continuity of the probabilistic safety-critical controller and give formal arguments about the existence and uniqueness of closed-loop executions of the system under our proposed safe online learning algorithm. In turn, this allows us to provide the main theoretical result of this paper (Theorem \ref{thm:main_theorem}), establishing safety in terms of \textit{set invariance with high probability}, even in regions \revision{where the prior data and the model knowledge are limited.} To our knowledge, this is the first work in the area of CBFs applied to systems with uncertain dynamics that collects data online and provides recursive feasibility guarantees of the CBF-based safe controller.

\revision{To sum up, the main components of this article are: 
\begin{enumerate}[leftmargin=1.25em, labelindent=\parindent, listparindent=\parindent, labelwidth=0pt]
    \item A concise summary of the necessary and sufficient conditions for the \textit{pointwise} feasibility of the GP-CBF-SOCP controller, originally presented in \cite{castaneda2021pointwise}, with improved notations and explanations to provide readers with intuition behind our analysis (Section \ref{sec:05feasibility}).
    \item The design of a new event-triggered online safe learning algorithm, derived from the feasibility analysis, which takes safe exploration actions whenever necessary to ensure the \textit{recursive} feasibility of the controller (Section \ref{subsec:main-algorithm}).
    \item A formal proof of the probabilistic forward invariance of the CBF zero-superlevel set under the proposed control algorithm (Section \ref{subsec:main-theory-part}).
\end{enumerate}
}

\subsection{Notations}
\revision{
\small
\noindent$B$: CBF (Definition \ref{def:cbf})

\noindent $\tilde{\dot{B}}$: estimate of the CBF derivative based on the nominal model (\eqref{eq:Btilde})

\noindent$\mathbb{D}_N$: dataset containing $N$ points for the GP regression

\noindent $e_\dagger(x|\mathbb{D}_N)$: the eigenvector associated with $\eigval(x|\mathbb{D}_N)$

\noindent$f, g$: true plant vector fields in \eqref{eq:system}

\noindent$\tilde{f}, \tilde{g}$: nominal model vector fields in \eqref{eq:nominal-model}

\noindent$\feasmat$: feasibility tradeoff matrix \eqref{eq:defineF}

\noindent$k_c$: Affine Dot Product (ADP) kernel (Definition \ref{def:adpkernel})

\noindent$\widehat{L_f B}(x|\mathbb{D}_{N})$, $\widehat{L_g B}(x|\mathbb{D}_{N})$: GP mean-based estimate of the Lie derivatives of $B$ in \eqref{eq:LfBhat} and \eqref{eq:LgBhat}.

\noindent$m$: control input dimension

\noindent$n$: state dimension

\noindent$u$: control input

\noindent$u_{\text{ref}}$: reference controller (to achieve a user-defined task)

\noindent$x$: state

\noindent$\mathcal{X}$: state domain

\noindent$\mathcal{X}_{\text{safe}} := \{ x \in \mathcal{X}: B(x) \geq 0 \}$

\noindent$y=[1, u^\top]^\top$ augmented control vector.

\noindent$z_B$: noisy measurement of $\Delta_B(x, u)$

\noindent$\beta$: constant in Lemma \ref{lemma:DeltaUCB} for the probabilistic bound of $\Delta_B$

\noindent$\gamma$: comparison function in Definition \ref{def:cbf}

\noindent$\delta$: probability level of the CBF chance constraint (Lemma \ref{lemma:DeltaUCB})

\noindent$\Delta_B$: model uncertainty term in \eqref{eq:mismatch_cbf}

\noindent$\epsilon$: threshold for $\lambda_\dagger$ in Algorithm \ref{algo:safelearning} ($\lambda_\dagger < -\epsilon$)

\noindent $\eigval(x|\mathbb{D}_N)$: minimum eigenvalue of $\feasmat$ in Lemma \ref{lemma:sufficient}.

\noindent$\gpmuB (x,u|\mathbb{D}_{N})$: GP posterior mean in \eqref{eq:mu_feas}

\noindent$\sigma_B^{2}(x,u|\mathbb{D}_{N})$: GP posterior variance in \eqref{eq:var_feas}

\noindent$\sigma_n^2$: measurement noise variance

\noindent$\gpGramB(x|\mathbb{D}_N)$: Gram matrix in the GP posterior variance \eqref{eq:var_feas}

\noindent $\gpGramB^{1/2}, \gpGramfB^{1/2}, \gpGramUB^{1/2}$: matrix square root of $\gpGramB$, and its drift and control relevant components (\eqref{eq:upper_right_uncertainty_sqrt}, \eqref{eq:lower_right_uncertainty_sqrt})

\noindent$\tau_{\text{max}}$: maximum time of existence and uniqueness of $x(t)$.

\normalsize
} 

\section{Problem Statement}
\label{sec:02background}
Throughout the paper we consider a control-affine nonlinear system of the following form:
\vspace{-3pt}
\begin{equation}\label{eq:system}
    \dot{x} = f(x) + g(x)u,
    \vspace{-3pt}
\end{equation}
where $x \in \mathcal{X} \subset \R^n$ is the state and $u \in \mathbb{R}^m$ is the control input. Many important classes of real-world systems, such as those with Lagrangian dynamics, can be represented in this form. We assume that $f: \mathcal{X} \to \R^n$ and $g: \mathcal{X} \to \R^{n\times m}$ are locally Lipschitz continuous. We will call system \eqref{eq:system} the \textit{true plant}. The problem addressed in this paper is how to guarantee the safety of the true plant \eqref{eq:system} when its dynamics $f$ and $g$ are unknown, 
while trying to accomplish a desired task.
Our proposed method will tackle this problem using real-time data and an approximate \textit{nominal model} of the system's dynamics, with $\tilde{f}: \mathcal{X} \to \R^n$ and $\tilde{g}: \mathcal{X} \to \R^{n\times m}$,
\vspace{-3pt}
\begin{equation}\label{eq:nominal-model}
    \dot{x} = \tilde{f}(x) + \tilde{g}(x)u.
    \vspace{-3pt}
\end{equation}

\subsection{Safety with Perfectly Known Dynamics}

We first formalize the notion of safety. In particular, we say that a control law $u: \mathcal{X} \to \R^m$ guarantees the safety of system \eqref{eq:system} with respect to a \textit{safe set} $\mathcal{X}_{\text{safe}} \subset \mathcal{X}$, if the set $\mathcal{X}_{\text{safe}}$ is forward invariant under the control law $u$.

\begin{definition}[Safety as forward invariance] A set $\mathcal{X}_{\text{safe}}$ is forward invariant under a control law $u: \mathcal{X} \to \R^m$
if, for any $x_0 \in \mathcal{X}_{\text{safe}}$, the closed-loop solution $x(t)$ of system \eqref{eq:system} under $u$ remains in $\mathcal{X}_{\text{safe}}$ for all $ t \in [0, \tau_{max})$. Here, $\tau_{max}$ is the maximum time of existence and uniqueness of $x(t)$, which exists and is strictly greater than zero if the control law $u$ is locally Lipschitz continuous in $x$ \cite{khalil2002nonlinear}.
\end{definition}

\revision{
\begin{remark}
\label{rmk:tau_max}
Throughout this article, we will say that a control policy is safe if it satisfies the above invariance condition over the entire period of existence and uniqueness of solutions of the closed-loop system, defined by $\tau_{max}$. If the closed-loop system is forward complete ($\tau_{max}=\infty$), the invariance condition is guaranteed indefinitely.
\end{remark}

Next, we introduce the concept of CBF and its usage for ensuring safety as defined above.
} 

\begin{definition}[Control Barrier Function \cite{ames2017cbf}]
\label{def:cbf} Let $\mathcal{X}_{\text{safe}} = \{ x \in \mathcal{X}: B(x) \geq 0 \}$ be the zero-superlevel set of a continuously differentiable function $B: \mathcal{X} \to \R$. Then, $B$ is a \emph{Control Barrier Function} (CBF) for system \eqref{eq:system} if there exists an extended class $\mathcal{K}_\infty$ function\footnote{\revision{A function $\gamma:\R\rightarrow\R$ is an extended class $\mathcal{K}_\infty$ function if it is strictly increasing and $\gamma(0) = 0$ \cite{ames2017cbf}.}} $\gamma$ such that for all $x \in \mathcal{X}$ the following holds:
\begin{equation}
    \label{eq:cbf_def}
    \sup_{u \in \R^m} \underbrace{L_f B(x) + L_g B(x) u}_{= \dot{B}(x,u)} + \gamma (B(x)) \geq 0,
\end{equation}
where $L_f B(x)\!\coloneqq\!\nabla B(x)\!\cdot\!f(x)$ and $L_g B(x)\!\coloneqq\!\nabla B(x)\!\cdot\!g(x)$ are the Lie derivatives of $B$ with respect to $f$ and $g$.
\end{definition}

The following result states that the existence of a CBF guarantees that the control system is safe.

\begin{lemma}
\label{lemma:CBF-Invariance}
\cite[Cor.~2]{ames2017cbf}
Let system \eqref{eq:system} admit a CBF $B: \mathcal{X} \to \R$. Let $\mathcal{X}_{\text{safe}}= \{ x \in \mathcal{X}: B(x) \geq 0 \}$ be its associated safe set, with boundary $\partial \mathcal{X}_{\text{safe}}= \{ x \in \mathcal{X}: B(x) = 0 \}$. If for all $x \in \partial \mathcal{X}_{\text{safe}}$ it holds that $\nabla B(x) \neq 0$, then any Lipschitz continuous control law $u: \mathcal{X} \to \R^m$ satisfying
\begin{equation}
    \label{eq:cbf_condition}
    u(x) \in \{u\in \R^m : \dot{B}(x,u) + \gamma\left(B(x) \right) \geq 0\}
\end{equation}
renders the set $\mathcal{X}_{\text{safe}}$ forward invariant.
\end{lemma}

Given a safety-agnostic reference controller $u_{\text{ref}}: \mathcal{X} \to \R^m$, the condition in \eqref{eq:cbf_condition} can be used to formulate a minimally-invasive safety-filter \cite{ames2017cbf}:

{\small
\HRule
\noindent
\vspace{-1em}
\begin{subequations}
\label{eq:cbf-qp-all}
\begin{align}
& \hspace{-2.5em} \textbf{CBF-QP}: \;\; && u^{*}(x) = \underset{u\in \R^{m}}{\argmin}  \quad \norm{u-u_{\text{ref}}(x)}_2^2 \label{eq:cbf-qp-cost}\\
& \quad \quad \text{s.t.} && L_f B(x) + L_g B(x)u + \gamma (B(x)) \geq 0, \label{eq:cbf-constraint} \vspace{-.5em}
\end{align}
\end{subequations}
\hrule
\normalsize
\vspace{2mm}}
\noindent which is a quadratic program (QP). \eqref{eq:cbf-qp-all} is solved pointwise in time to obtain a safety-critical control law $u^* : \mathcal{X} \to \R^m$ that only deviates from the reference controller $u_{\text{ref}}$ when safety is compromised. However, note that this optimization problem requires perfect knowledge of the dynamics of the system, since the Lie derivatives of $B$ with respect to the true dynamics $f$ and $g$ appear in the constraint.
Therefore, throughout this paper we will refer to \eqref{eq:cbf-constraint} as the \textit{true CBF constraint}. Note that this constraint is affine in the control input.

In \cite[Thm. 8]{xu2015robustness}, it is shown that if $u_{\text{ref}}$ and $\gamma$ are Lipschitz continuous functions, $B$ has a Lipschitz continuous gradient and if it satisfies the relative degree one condition in $\mathcal{X}$, i.e., $L_gB(x) \neq 0\ \forall x \in \mathcal{X}$; then the CBF-QP of \eqref{eq:cbf-qp-all} yields a locally Lipschitz control policy, therefore guaranteeing the forward invariance of $\mathcal{X}_{\text{safe}}$ by Lemma \ref{lemma:CBF-Invariance}.

\subsection{Safety under Model Uncertainty}
\label{subsec:problem-setup}
Now, we consider the case where the true dynamics \eqref{eq:system} are uncertain and only a nominal model \eqref{eq:nominal-model} is available. \revision{Moreover, we might not have any initial dataset containing previous trajectories of the true plant.} Instead, the system must autonomously reason about what data it needs to collect online in order to stay safe with a high probability.

\begin{problem}
\label{prob:problem-statement}
For a given safe set $\mathcal{X}_{\text{safe}} = \{x \in \mathcal{X}: B(x)\geq 0 \}$ and nominal dynamics model \eqref{eq:nominal-model}, design a data collection strategy and a data-driven control law $u : \mathcal{X} \to \R^m$ that together render the set $\mathcal{X}_{\text{safe}}$ forward invariant for system \eqref{eq:system} with a high probability, i.e.,

\vspace{-10pt}
{\footnotesize
\begin{equation*}
\label{eq:prob-invariance}
\mathbb{P}\{\ \forall x_0 \in \mathcal{X}_{\text{safe}},\ x(0) = x_0 \implies x(t) \in \mathcal{X}_{\text{safe}}, \forall t \in [0,\tau_{max})\ \}\geq 1-\delta, 
\vspace{-3pt}
\end{equation*}}

\noindent where 
$\delta \in (0,1]$ is a user-defined risk tolerance.
\end{problem}

\revision{The primary assumption we make is that we have access to a CBF for the set $\mathcal{X}_{\text{safe}}$ that is valid for the true plant:
\begin{assumption}
\label{assumption:valid-cbf}
$B: \mathcal{X} \to \R$ is a valid CBF for the true plant \eqref{eq:system}, with zero-superlevel set $\mathcal{X}_{\text{safe}} = \{ x \in \mathcal{X}: B(x) \geq 0 \}$, i.e. there exists a control policy $u: \mathcal{X} \to \R^m$ satisfying \eqref{eq:cbf_def} for the true plant \eqref{eq:system}.
\end{assumption}}

\revision{This assumption implies that a control policy exists (albeit unknown) to keep the set $\mathcal{X}_{\text{safe}}$ forward invariant for the true plant \eqref{eq:system}.
However, even when a valid CBF is available, obtaining such a control policy is not straightforward with no accurate access to $f$ and $g$, which constitutes the main challenge we address. Assumption \ref{assumption:valid-cbf} is also present in the prior works that aimed to solve similar problems as ours \cite{taylor2020cbf,taylor2020towards,choi2020reinforcement,dhiman2021control,greeff2021learning}.}

\revision{
\begin{remark} 
\label{remark:learning-based-cbf}
\textit{(CBFs for uncertain systems)}
    Designing effective CBFs (those that are not overly conservative) for uncertain systems is an active area of research \cite{dawson2022safe, qin2022sablas, jagtap2020control, lindemann2021learning, jin2020neural, nejati2023data}. Our contribution is relevant but runs parallel to such line of research, since even with a valid CBF, finding a control action that satisfies the constraint \eqref{eq:cbf-constraint} is in itself a challenge for uncertain dynamics. In fact, the efforts in designing the CBF and in formulating the control input constraint for uncertain systems complement each other, as the safety for uncertain systems can only be guaranteed when both aspects are integrated. 

    Different works summarized in \cite{cohen2024safety} suggest that CBFs designed for reduced-order models can be used as (possibly conservative) CBFs for full-order systems. This indicates that our assumption of having a CBF that is valid for the true plant is not overly restrictive and just requires the nominal model in \eqref{eq:nominal-model} to capture relevant structural properties of the true plant, as discussed in Sections III-V of \cite{cohen2024safety}.
    Furthermore, due to the inherent robustness properties of CBFs, with conservative estimates of the bounds of the uncertainty of the nominal model, a robust CBF based on the nominal model can be designed and be used for the true plant \cite{xu2015robustness, kolathaya2018input}. Using the nominal model is also known to be a reasonable procedure for feedback linearizable systems whose relative degree is known. As long as the relative degree structure remains identical for the nominal model and the true plant, a CBF valid for the nominal model is also valid for the true plant \cite{taylor2020towards}.
\end{remark}
} 

\revision{The CBF constraint \eqref{eq:cbf-constraint} is expressed in terms of the true plant but with the uncertain vector fields $f, g$. We can equivalently express this in terms of the nominal model, and a scalar uncertain term $\Delta_B$:}
\begin{equation}
    \label{eq:cbf-constraint-uncertainty}
    L_{\tilde{f}}B(x) + L_{\tilde{g}}B(x) u + \Delta_B(x, u) + \gamma(B(x)) \ge 0,
\end{equation}
\noindent where $L_{\tilde{f}}B$ and $L_{\tilde{g}}B$ are the Lie derivatives of $B$ computed using the nominal dynamics model \eqref{eq:nominal-model}, and the uncertain term $\Delta_B$ is defined for each $x \in \mathcal{X}$, $u \in \R^m$ as
\begin{equation}
    \Delta_B(x, u) := (L_f B \!-\!L_{\tilde{f}}B)(x) + (L_g B\!-\!L_{\tilde{g}}B)(x)u. \label{eq:mismatch_cbf}
\end{equation}
In the next section, we will present a method to estimate the function $\Delta_B$ using data from the true plant and Gaussian Process regression. By doing so, it is possible to formulate a probabilistic version of the optimization problem \eqref{eq:cbf-qp-all} that takes into account the current best estimate of the term $\Delta_B$ and the estimation uncertainty. Note that learning $\Delta_B$ is advantageous rather than learning the full dynamics of the system since $\Delta_B$ is a scalar function. Indeed, this function condenses all the safety-relevant model uncertainty into a scalar.

\section{Gaussian Process Regression-based Probabilistic Safety Filter}
\label{sec:03gp}
In this section, we first present background knowledge from \cite{GPCLFSOCP} on the use of
GP regression
\revision{to obtain predictions of $\Delta_B(x,u)$} with high probability error bounds. \revision{Then, we introduce the design of the SOCP-based controller, the GP-CBF-SOCP, that provides a probabilistic safety guarantee for the uncertain system when the optimization is feasible.}

\subsection{Gaussian Process Regression for Learning $\Delta_B(x, u)$}

A Gaussian Process (GP) is a random process for which any finite collection of samples have a joint Gaussian distribution. It is characterized by its mean $q: \gpdomain\rightarrow\R$ and covariance (or kernel) $k: \gpdomain\times\gpdomain\rightarrow\R$ functions, where $\gpdomain$ is the input domain of the process. In GP regression, an unknown function $h(\cdot)$ is assumed to be a sample from a GP, and through a set of $N$ noisy measurements $\mathbb{D}_N=\{\gpvar_j,h(\gpvar_j)+\epsilon_j\}_{j=1}^{N}$, a prediction of $h(\cdot)$ at an unseen query point $\gpvar_*$ can be derived from the joint distribution of $[h(\gpvar_1), \cdots, h(\gpvar_N), h(\gpvar_{*})]^T$ conditioned on the dataset $\mathbb{D}_N$. $\epsilon_j \sim \mathcal{N}(0, \sigma_{n}^2)$ is white measurement noise, with $\sigma_{n}>0$. \revision{We use the zero mean function, $q \equiv 0$, since the prior information (based on the nominal model) is already captured in $L_{\tilde{f}}B(x), L_{\tilde{g}}B(x)$ in \eqref{eq:cbf-constraint-uncertainty}.}

Given the dataset $\mathbb{D}_N$, the mean and variance of the prediction from the GP regression are given as
\begin{equation}
\label{eq:gpposteriormu}
        \mu(x_*|\mathbb{D}_N) = \mathbf{z}^T (K + \sigma_n^2 I )^{-1} K_{*}^{T}, \vspace{-.5em}
\end{equation}
\begin{equation}
\label{eq:gpposteriorsigma}
    \sigma^{2}(x_*|\mathbb{D}_N) = k\left(\gpvar_{*}, \gpvar_{*}\right)-K_{*}  (K + \sigma_n^2 I )^{-1} K_{*}^{T}, \vspace{-.3em}
\end{equation}
where $K\in\R^{N\times N}$ is the GP Kernel matrix, whose $(i, j)^{th}$ element is $k(\gpvar_i, \gpvar_j)$, $K_{*}=[k(\gpvar_*, \gpvar_1),\ \cdots \ ,k(\gpvar_*, \gpvar_N)]\in\R^N$, and $\mathbf{z}\in \R^N$ is the vector containing the noisy measurements of $h$, $z_j:=h(\gpvar_j)+\epsilon_j$.

The choice of kernel $k$ determines properties of the target function $h$, like its smoothness and signal variance. Moreover, the kernel can be used to express prior structural knowledge of the the target function \cite{duvenaud2014automatic}.
\revision{In our case, we want to exploit the fact that our target function $\Delta_B$ from \eqref{eq:mismatch_cbf} is control-affine}.
We first rewrite \eqref{eq:mismatch_cbf} as
\begin{equation}
\label{eq:delta_with_phi}
    \Delta_B(x, u) = \Phi_B(x) \cdot \left[\begin{array}{c} 1 \\ u \end{array}\right],
\end{equation}
where $\Phi_B(x) := \left[
         L_f B(x) \!-\!L_{\tilde{f}}B(x)  \quad
         L_g B(x) \!-\!L_{\tilde{g}}B(x)    \right]$.
We then consider a GP regression for $\Delta_B$, with domain $\gpaugdomain := \mathcal{X} \times \R^{m+1}$, where $\R^{m+1}$ is the space of $\augu\coloneqq[1, u^T]^T$. For this, we use the compound kernel presented in \cite{GPCLFSOCP} to exploit the structure of \eqref{eq:delta_with_phi}:

\begin{definition}[Affine Dot Product compound kernel \cite{GPCLFSOCP}] \label{def:adpkernel}
Define $k_{c}:\gpaugdomain \times \gpaugdomain \rightarrow \R$ given by
\vspace{-1em}

\small
\begin{equation}
    k_{c}\left(\left[\begin{array}{c} x \\ y \\\end{array}\right], \left[\begin{array}{c} x' \\ y' \\\end{array}\right]\right) := y^T Diag([k_1(x, x'), \cdots, k_{m+1}(x, x')]) y'
\label{eq:adpkernel}
\end{equation}
\normalsize

\noindent as the \emph{Affine Dot Product} (ADP) compound kernel of $(m\!+\!1)$ individual kernels $k_1,\ldots ,$ $k_{m+1}:\mathcal{X}\times \mathcal{X} \rightarrow \R$.
\end{definition}

\revision{Using the ADP compound kernel as the covariance function for the GP regression, equations \eqref{eq:gpposteriormu} and \eqref{eq:gpposteriorsigma} take the following form for the mean and variance of the GP prediction for $\Delta_B$ at a query point $(x_{*}, \augu_{*})$:
\vspace{-3pt}
\begin{equation}
\label{eq:mu_adp}
    \mu_B(x_*,y_*|\mathbb{D}_N) = \underbrace{\mathbf{z}^T (K_c + \sigma_n^2 I )^{-1} K_{*Y}^T}_{=:\ m_B(x_*|\mathbb{D}_N)^T} y_{*},
\end{equation}
\vspace{-5pt}
{\small
\begin{equation}
\label{eq:sigma_adp}
    \sigma_B^{2}(x_*,y_*|\mathbb{D}_N) \!= \!y_{*}^{T}\!\underbrace{\left(K_{**} \!-\!K_{*Y}\!(K_c + \sigma_n^2 I )^{-1} K_{*Y}^{T}\! \right)}_{=:\ \Sigma_B(x_*|\mathbb{D}_N)}\!y_{*},
\end{equation}
\vspace{-6pt}}

\noindent where $K_c\in\R^{N\times N}$ is the Gram matrix of $k_c$ in \eqref{eq:adpkernel} for the training data inputs ($X,Y$), $K_{**} = \!Diag \text{ $\left([k_1(x_{*}, x_{*}), \ldots, k_{m+1}(x_{*}, x_{*})]\right)$} \in \R^{(m+1)\times(m+1)}$, and $K_{*Y}\in\R^{(m+1)\times N}$ is given by

{\small
\vspace{-10pt}
\begin{equation*}
    K_{*Y}\!=\!\begin{bmatrix} K_{1*} \\ \vdots \\K_{(m+1)*}
    \end{bmatrix}\!\circ\!Y,\  \text{with}\;\;K_{i*}\!=\![k_i(x_{*}, x_1),\cdots, k_i(x_{*}, x_N)].
\end{equation*}}

\noindent Here, $\circ$ denotes the element-wise product. The prediction of $\Delta_B(x_*,u_*)$ has a mean function \eqref{eq:mu_adp} that is affine in the control input $u_*$ and a variance function \eqref{eq:sigma_adp} that is quadratic in $u_*$. \revision{This arises from the control-affine nature of the dynamics and the use of the ADP compound kernel. This structure leads to the convexity of the optimization in the safety filter, which will be introduced in Section \ref{sec:04controllers}.} 
} 

We now revisit Theorem 2 of \cite{GPCLFSOCP} to construct a probabilistic bound on the true value of $\Delta_B$ from the GP prediction. To provide guarantees about the unknown function at an arbitrary point that is not necessarily part of the data, a set of assumptions, initially established in \cite[Thm. 6]{gpucb}, is required. In particular, the target function is required to belong to the Reproducing Kernel Hilbert Space (RKHS, \cite{wendland2004scattered}) $\mathcal{H}_k(\gpaugdomain)$ of the chosen kernel, and have a bounded RKHS norm $\norm{\cdot}_{k}$.
An RKHS, denoted as $\mathcal{H}_k(\gpaugdomain)$, is characterized by a specific positive-definite kernel $k$. The kernel $k$ evaluates whether a member of $\mathcal{H}_k(\gpaugdomain)$ satisfies a specific property, namely a ``reproducing" property: the inner product between any member function $h\in\mathcal{H}_k(\mathcal{\gpaugdomain})$ and the kernel $k(\cdot, \bar{x})$ should reproduce $h$, i.e., $\langle h(\cdot), k(\cdot, \bar{x}) \rangle_{k} = h(\bar{x}), ~\forall \bar{x}\in \gpaugdomain$. Furthermore, the RKHS norm $\norm{h}_{k}:=\sqrt{\langle h, h \rangle}_{k}$
is a measure of how ``well-behaved''\footnote{$\norm{h(\bar{x})-h(\bar{x}')}_{2}\le\norm{h}_{k}\norm{k(\bar{x},\cdot)-k(\bar{x}',\cdot)}_{k} \; \forall \bar{x}, \bar{x}' \in \mathcal{\gpaugdomain}$} the function $h\in\mathcal{H}_k(\mathcal{\gpaugdomain})$ is.

\begin{lemma}
\label{lemma:DeltaUCB} \revision{(Probabilistic bound of $\Delta_B(x, u)$)} \cite[Theorem 2]{GPCLFSOCP} Consider $m\!+\!1$ bounded kernels $k_{i}$, for $i\!=\!1,\ldots,(m+1)$. Assume that the $i$th element of $\Phi_B$ is a member of $\mathcal{H}_{k_{i}}$ with bounded RKHS norm, for $i\!=\!1,\ldots,(m+1)$. Moreover, assume that we have access to a dataset $\mathbb{D}_N=\{(\gpvar_j,u_j),\ \Delta_B(\gpvar_j,u_j)+\epsilon_j\}_{j=1}^{N}$ of $N$ noisy measurements, and that $\epsilon_j$ is zero-mean and uniformly bounded by $\sigma_{n} > 0$. Let $\beta \coloneqq \left(2\eta^2 + 300 \kappa_{N+1} \ln^3((N+1)/\delta)\right)^{0.5}$,
with $\eta$ the bound of $\norm{\Delta_B}_{k_c}$, $\kappa_{N+1}$ the maximum information gain after getting $N+1$ data points, and $\delta \in (0,1)$. Let $\mu_B$ and $\sigma^2_B$ be the mean \eqref{eq:mu_adp} and variance \eqref{eq:sigma_adp} of the GP regression for $\Delta_B$, using the ADP compound kernel $k_c$ of $k_1,\ldots,k_{m+1}$, at a query point $(x_{*},\ \augu_{*}=[1, u_{*}^T]^T)$, where $x_*$ and $\augu_{*}$ are elements of bounded sets $\mathcal{X}\subset\R^{n}$ and $\augU\subset\R^{m+1}$, respectively. Then, the following holds:
\vspace{-3pt}
\begin{multline}
\label{eq:sigmaUCB}
\hspace{-10pt}
    \mathbb{P}\bigg\{\ \bigg| \mu_B(x_*,y_*|\mathbb{D}_N) - \Delta_B(x_{*},u_{*}) \bigg| \leq \beta \sigma_B(x_*,y_*|\mathbb{D}_N),\\ \forall N \geq 1,\ \forall x_{*} \in \mathcal{X},\ \forall \augu_*=[1,u_{*}^T]^T \in \augU\bigg\}\geq 1-\delta.
\vspace{-3pt}
\end{multline}
\end{lemma}

\subsection{Probabilistic Safety Filter (GP-CBF-SOCP)}
\label{sec:04controllers}
We now make use of the probability bound given by Lemma \ref{lemma:DeltaUCB} to build an uncertainty-aware CBF chance constraint that can be incorporated in a minimally invasive probabilistic safety filter.
Let us take the lower bound of \eqref{eq:sigmaUCB} and note that $\dot{B}(x,u) = \tilde{\dot{B}}(x,u) + \Delta_B(x,u)$, 
\revision{ where 
\vspace{-0.5em}
\begin{equation}
\label{eq:Btilde}
\tilde{\dot{B}}(x,u) =  L_{\tilde{f}}B(x) + L_{\tilde{g}}B(x) u    
\vspace{-0.5em}
\end{equation}
is the CBF derivative computed using the nominal dynamics model \eqref{eq:nominal-model}. }
Then, as a result of Lemma \ref{lemma:DeltaUCB}, the following inequality holds with a compound probability (for all $x \!\in\!\mathcal{X}$, $[1\ u^T]^T \in \augU$ and $N\geq 1$) of at least $1-\delta$:
{\setlength{\belowdisplayskip}{3pt}
\setlength{\abovedisplayskip}{3pt}
\begin{align}
\vspace{-8pt}
\label{eq:B_UCB}
    \dot{B}(x,u) & \geq \tilde{\dot{B}}(x,u) + \gpmuB(x,u|\mathbb{D}_N) - \beta \gpsigmaB(x,u|\mathbb{D}_N).
    \vspace{-2pt}
\end{align}
}

Inequality \eqref{eq:B_UCB} gives a worst-case high-probability bound of the true CBF derivative. \revision{The bound is constructed based on the nominal model ($\tilde{\dot{B}}$) and the GP regression from data ($\gpmuB, \gpsigmaB$).}
We use this lower bound of the CBF derivative to construct a probabilistically robust CBF chance constraint that can be evaluated without explicit knowledge of the dynamics of the true plant, and we incorporate it in a chance-constrained reformulation of the CBF-QP safety filter:

{\small
\vspace{2mm}
\hrule
\vspace{2mm}
\begin{subequations}
\label{eq:gp-cbf-socp}
\begin{align}
& \hspace{-1em} \textbf{GP-CBF-SOCP}: \;\;\; u^{*}(x) =   \underset{u\in \R^{m}}
{\argmin}\norm{u-u_{\text{ref}}(x)}_2^2 \ \  \text{s.t.}
 \\
& \quad \;\; \tilde{\dot{B}}(x,u)\!+\!\gpmuB(x,u|\mathbb{D}_N)\!+\!\gamma (B(x)) \geq \!\beta \gpsigmaB(x,u|\mathbb{D}_N). \label{eq:socp-cbf-constraint}
\end{align}
\end{subequations}
\hrule
\vspace{2mm}}

This problem is solved at each timestep in real-time to obtain a safety-filtered control law $u^* : \mathcal{X} \to \R^m$ that only deviates from the reference $u_{\text{ref}}$ when safety is compromised for the desired probability bound of $1- \delta$.

The linear and quadratic structures of the mean \eqref{eq:mu_adp} and variance \eqref{eq:sigma_adp} of the GP prediction of $\Delta_B$, respectively, result in \eqref{eq:gp-cbf-socp} being a second-order cone prgram (SOCP), as shown in Theorem \ref{theorem:SOCP} below. Therefore, by exploiting the control-affine structure of the system during the GP regression, we obtain a convex optimization problem that can be solved at high-frequency rates when using modern solvers.

\begin{theorem}
\label{theorem:SOCP}
\revision{\!(\!\cite[Theorem 3]{GPCLFSOCP})} For an unknown control-affine system \eqref{eq:system} with associated CBF $B$, let $\gpmuB$ and $\gpsigmaB^2$ be the mean and variance functions of the GP prediction of $\Delta_B$ using the ADP compound kernel from Definition \ref{def:adpkernel}. Then, the probabilistic safety filter of \eqref{eq:gp-cbf-socp} is an SOCP.
\end{theorem}

\revision{Theorem \ref{theorem:SOCP} is proved in \cite{GPCLFSOCP} by rewriting \eqref{eq:gp-cbf-socp} in the standard SOCP form. Here, we provide new notations necessary
for the subsequent parts of the paper.
From \eqref{eq:mu_adp} and \eqref{eq:sigma_adp}, we have {\setlength{\belowdisplayskip}{3.5pt}
\setlength{\abovedisplayskip}{3pt}
\begin{align}
    \hspace{-3pt}\gpmuB (x,u|\mathbb{D}_N) &= m_{B}(x|\mathbb{D}_N)^T [1 \ u^T]^T\!, \label{eq:mu_feas}\\ 
    \gpsigmaB^2(x,u|\mathbb{D}_N) &= [1 \;\; u^T] \gpGramB(x|\mathbb{D}_N)[1 \ u^T]^T,\label{eq:var_feas}
\end{align}}
and we can write
\vspace{-5pt}
\begin{equation}
\label{eq:cone_feas}
    \gpsigmaB(x, u |\mathbb{D}_N) = \norm{\gpGramB^{1/2}(x|\mathbb{D}_N) \begin{bmatrix}1 \\ u\end{bmatrix}}_2,
\end{equation}
where $\gpGramB^{1/2}(\cdot)\in\R^{(m+1)\times(m+1)}$ is the matrix square root of $\gpGramB(\cdot)$. Note that $\gpGramB(x|\mathbb{D}_N)$ is positive definite.

We now define:

{\small \begin{align}
\hspace{-5pt}\widehat{L_f B}(x|\mathbb{D}_N) &:= L_{\tilde{f}}B(x) + m_{B}(x|\mathbb{D}_N)_{[1]} \in \R, \label{eq:LfBhat}\\
\hspace{-10pt}\socc &:= L_{\tilde{g}}B(x) + m_{B}(x|\mathbb{D}_N)_{[2:(m+1)]}^T \in \R^{1\times m}, \label{eq:LgBhat}\\
\hspace{-5pt}\gpGramfB^{1/2}(x|\mathbb{D}_N) &:= \gpGramB^{1/2}(x|\mathbb{D}_N)_{[1:(m+1)],[1]} \in \R^{m+1}, \label{eq:upper_right_uncertainty_sqrt}\\
\hspace{-5pt}\gpGramUB^{1/2}(x|\mathbb{D}_N) &:= \gpGramB^{1/2}(x|\mathbb{D}_N)_{[1:(m+1)],[2:(m+1)]} \in \R^{(m+1)\times m},
\label{eq:lower_right_uncertainty_sqrt}
\end{align}
}

\noindent where numerical subscripts denote elements of vectors or matrices. Intuitively, $\widehat{L_f B}(x|\mathbb{D}_N)$ and $\socc$ are the mean predictions of the true plant's $L_fB(x)$ and $L_gB(x)$, respectively. $\gpGramfB^{1/2}(x|\mathbb{D}_N)$ and $\gpGramUB^{1/2}(x|\mathbb{D}_N)$ correspond to the components of the uncertainty matrix $\gpGramB^{1/2}(x|\mathbb{D}_N)$ in \eqref{eq:cone_feas}. 

With these notations, \eqref{eq:socp-cbf-constraint} can be rewritten in the standard second-order cone constraint form as:
\begin{equation*}
\begin{split}
\socc u + \Big(\socd\Big)\\  \hspace{20pt}\ge \beta \norm{\gpGramUB^{1/2}(x|\mathbb{D}_N) u + \gpGramfB^{1/2}(x|\mathbb{D}_N)}_2.
 \label{eq:constraint_std}
\end{split}    
\end{equation*}
}

\section{Analysis of Pointwise Feasibility}
\label{sec:05feasibility}

The GP-CBF-SOCP, if feasible, is guaranteed to provide $u$ that satisfies the true CBF constraint \eqref{eq:cbf-constraint} with high probability. \revision{However, unlike QPs, SOCPs with even only a single constraint can be infeasible. Our SOCP can be infeasible when the prediction uncertainty is significant enough to obstruct the discovery of a suitable control input that ensures safety. In other words, \eqref{eq:gp-cbf-socp} is infeasible when the uncertainty ($\gpsigmaB$) is dominant in the CBF chance constraint \eqref{eq:socp-cbf-constraint}. This is in} \revision{contrast to the uncertainty-free case, where the CBF-QP \eqref{eq:cbf-qp-all} is guaranteed to always be feasible by the definition of CBF. It is therefore essential to study under which conditions the GP-CBF-SOCP becomes infeasible, as safety could be compromised in those cases. The result in this section is a summary of the feasibility analysis conducted in the conference version of this work\cite{castaneda2021pointwise} with improved notations.}

The first feasibility result we present is a
necessary condition for pointwise feasibility of the GP-CBF-SOCP:

\begin{lemma}
\label{lemma:necessary}
\revision{\!(\!\cite[Lemma 2]{castaneda2021pointwise} Necessary condition)} For a given dataset $\mathbb{D}_N$, if the GP-CBF-SOCP \eqref{eq:gp-cbf-socp} is feasible at a point $x \in \mathcal{X}$, then it must hold that
\begin{equation}
\label{eq:necessary}
\scriptsize
     \begin{bmatrix}\socd \\ \socc^T\end{bmatrix}^T\!\gpGramB(x|\mathbb{D}_N)^{-1}\!\begin{bmatrix}\socd \\ \socc^T\end{bmatrix} \geq \beta^2,
\end{equation}
\revision{where $\beta$ is defined in Lemma \ref{lemma:DeltaUCB}.}
\end{lemma}

\begin{proof}
See Appendix \ref{subsec:proof-lemma-necessary}.
\end{proof}

To provide insight into this condition, for a given data set $\mathbb{D}_N$ and at a particular point $x \in \mathcal{X}$, note that the left-hand side of \eqref{eq:necessary} encodes a trade-off between the uncertainty matrix $\gpGramB(x|\mathbb{D}_N)$ and the mean prediction of the terms of the CBF constraint (as in the vector $[\socd, \ \socc]$).
\revision{The term $\socc$ reflects the mean prediction of how a control input $u$ can cause a change in the value of the CBF $B(x)$. The value of $B(x)$ can be regulated at state $x$ if the groundtruth $L_gB(x)$ is non-zero.} In this case, the control-relevant components of \eqref{eq:necessary} reveal that the necessary condition for pointwise feasibility is easily satisfied if the value of $\socc$ is dominant over the lower-right block of $\gpGramB(x|\mathbb{D}_N)$, which captures the growth of the prediction uncertainty with respect to $u$.

We next state the sufficient condition for pointwise feasibility of the GP-CBF-SOCP, which will be the foundation for the algorithm we present in the next section. 

\revision{We first note that the lower-right block of the uncertainty matrix $\gpGramB(x|\mathbb{D}_N)$, can be expressed as
\begin{multline}
\socAA := \gpGramB(x|\mathbb{D}_N)_{[2:(m+1)],[2:(m+1)]} = \\ \socA^T\ \socA \in \R^{m\times m},
\end{multline}
where $\gpGramUB^{1/2}$ is from \eqref{eq:lower_right_uncertainty_sqrt}.
} 
The sufficient condition originates from the following matrix, which we call the \emph{feasibility tradeoff matrix}:
\begin{equation}
\label{eq:defineF}
\feasmat \coloneqq \beta^2\gpGramUB(x|\mathbb{D}_N) - \socc^T\socc.
\end{equation}

The matrix $\feasmat$ encodes the tradeoff between the uncertainty of the prediction and the predicted safety. In fact, the first term in \eqref{eq:defineF}, $\beta^2\gpGramUB(x|\mathbb{D}_N)$, is a positive-definite matrix that informs about the uncertainty growth in each control direction; and the second term, $\socc^T\socc$, is a rank-one positive-semidefinite matrix capturing the mean prediction of the true plant's safest control direction $L_gB(x)$.

Note that the result of subtracting a rank-one positive-semidefinite matrix from a positive-definite matrix can have at most one negative eigenvalue. Our sufficient condition for pointwise feasibility states that if the rank-one subtraction term is strong enough to flip the sign of one of the eigenvalues of the uncertainty matrix, then there exists one feasible control input direction (defined by the corresponding eigenvector). Intuitively, along this control input direction, the controllability of the CBF is dominant over the growth of the prediction uncertainty. The following Lemma that formally presents the sufficient condition also provides an expression for such control input direction, that we call $\usafe(x)$, in a closed form:

\begin{lemma}
\label{lemma:sufficient}
\revision{\!(\!\!\cite[Lemma 3]{castaneda2021pointwise} Sufficient condition)} Given a dataset $\mathbb{D}_N$, for a point $x \in \mathcal{X}$, let $\eigval(x|\mathbb{D}_N)$ be the minimum eigenvalue of the feasibility tradeoff matrix $\feasmat$ \eqref{eq:defineF}, and $\eigvec(x|\mathbb{D}_N)$ be its associated unit eigenvector.
If $\eigval(x|\mathbb{D}_N)<0$, the GP-CBF-SOCP \eqref{eq:gp-cbf-socp} is feasible at $x$, and there exists a constant $\alpha_{min}> 0$ such that for any $\alpha > \alpha_{min}$,
\vspace{-3pt}
\begin{equation}
\label{eq:safe-direction}
    \usafe(x) = \alpha\ \text{sgn}\big(\socc \eigvec(x|\mathbb{D}_N)\big)\ \eigvec(x|\mathbb{D}_N)
\vspace{-3pt}
\end{equation}
is a feasible solution of \eqref{eq:gp-cbf-socp} at $x$.
\end{lemma}

\begin{proof}
    See Appendix \ref{subsec:proof-lemma-sufficient}.
\end{proof}

With this condition, a single scalar value, $\eigval$, being negative guarantees the feasibility of the GP-CBF-SOCP. This can be easily checked online before solving the optimization problem. Furthermore, for a particular state $x \in \mathcal{X}$, the value of $\eigval(x|\mathbb{D}_N)$ can be clearly associated with a notion of \emph{richness} of the dataset $\mathbb{D}_N$ in terms of characterizing safety---if it is negative, then there exists at least one control input direction which
keeps the system safe with high probability. This condition serves as the foundation for the safe online learning methodology that we present in Section \ref{sec:05safelearning}.

Lastly, we state the necessary and sufficient condition for pointwise feasibility of the GP-CBF-SOCP. This condition combines and generalizes Lemmas \ref{lemma:necessary} and \ref{lemma:sufficient}.

\begin{theorem}
\label{th:nec_and_suf}
\revision{(\!\cite[Theorem 2]{castaneda2021pointwise} Necessary \& sufficient condition)}
Given a dataset $\mathbb{D}_N$, for a point $x \in \mathcal{X}$, let $\eigval(x|\mathbb{D}_N)$ be the minimum eigenvalue of the feasibility tradeoff matrix $\feasmat$ defined in \eqref{eq:defineF}. Then, the GP-CBF-SOCP \eqref{eq:gp-cbf-socp} is feasible at $x$ if and only if condition \eqref{eq:necessary} is satisfied and one of the following cases holds:

    \noindent\textbf{1: }$\eigval(x|\mathbb{D}_N) <0$;
    
    \noindent\textbf{2: }$\eigval(x|\mathbb{D}_N)>0$,     
    and \vspace{-1em}
    {\small
    \begin{multline}
    \label{eq:case2-add}
    \hspace{-11pt}
        \socd-\\
        \socc \feasmat^{-1}\big[\beta^2\socA^T\socb -\\ \socc^T\big(\socd\big)\big]\!\geq\!0; \vspace{-1em}
    \end{multline}}
    
   \noindent \textbf{3: }$\eigval(x|\mathbb{D}_N) =0$,     
    and \vspace{-1em}
    {\small
    \begin{multline}
    \label{eq:case3-add}
    \hspace{-11pt}
        \socd - \\
        \socc\socAA^{-1}\socA^T \socb > 0. \vspace{-1em}
    \end{multline}}
\end{theorem}
\noindent Case 1 matches the sufficient condition of Lemma \ref{lemma:sufficient}, and it results in a hyperbolic feasible set. Cases 2) and 3) correspond to elliptic and parabolic feasible sets, respectively.
\begin{proof} See Appendix \ref{subsec:proof-theorem-necessary-sufficient}. \end{proof}

Theorem \ref{th:nec_and_suf} provides tight conditions that the available data $\mathbb{D}_N$ should satisfy in order to obtain probabilistic safety guarantees for systems with actuation uncertainty. 

\section{Probabilistic Safe Online Learning and Recursive Feasibility}
\label{sec:05safelearning}
\subsection{Proposed Safe Online Learning Strategy}
\label{subsec:main-algorithm}
In this section, we present the safe online learning algorithm that guarantees safety of the true plant \eqref{eq:system} with high probability \revision{through the recursive feasibility of \eqref{eq:gp-cbf-socp}}. Our algorithm uses the nominal dynamics model of \eqref{eq:nominal-model} and the online stream of data collected by the system as its state trajectory evolves with time, constructing a dataset $\mathbb{D}_N$ online. The design goal of our algorithm is to \revision{maintain} the feasibility of the GP-CBF-SOCP. Our main strategy to achieve this is to guarantee that the sufficient condition for pointwise feasibility of Lemma \ref{lemma:sufficient} always holds. By doing so, we ensure that there always exists a backup control direction $\usafe$ \eqref{eq:safe-direction} that can guarantee safety with a high probability.

\begin{remark}
\label{remark:notnecessary}
Note that Lemma \ref{lemma:sufficient} is only a sufficient condition for feasibility of the SOCP \eqref{eq:gp-cbf-socp}, and that the problem could be feasible at $x\in \mathcal{X}$ even when the condition $\eigval(x|\mathbb{D}_N) <0$ of Lemma \ref{lemma:sufficient} does not hold, \revision{as can be seen in Theorem \ref{th:nec_and_suf}.} However, if $\eigval(x|\mathbb{D}_N) <0$ does not hold, it means that there does not exist any control input direction that can serve as a backup safety direction, and the SOCP is feasible at $x$ only if the CBF constraint can be satisfied with $u \to 0$. This situation is not desirable since the system might later on move towards states where the true CBF constraint cannot be satisfied unless a control input is applied, in which case the SOCP would become infeasible.
\end{remark}

\revision{We now discuss the necessity of taking a control input action that may not help in achieving the task (as defined by $u_\text{ref}$) but might significantly help in securing the feasibility of the SOCP. Revisiting \eqref{eq:var_feas}, note that the matrix $\gpGramUB(x|\mathbb{D}_N)$ characterizes how the variance term $\gpsigmaB^2$ grows with respect to different control directions. If all the data points $(x_j,u_j)$ consisting $\mathbb{D}_N$ are control inputs from $u_\text{ref}$, then the uncertainty growth in the safe control direction $L_gB(x)$ can be potentially high, especially if $L_gB(x)$ and $u_\text{ref}$ have distinct directions. In this case, the condition $\eigval(x|\mathbb{D}_N) <0$ of Lemma \ref{lemma:sufficient} may not be satisfied and the SOCP might soon become infeasible.}

\revision{This situation happens if the system under the GP-CBF-SOCP is in regions where the CBF constraint is not active and the solution of \eqref{eq:gp-cbf-socp} is identical to $u_{\text{ref}}$. The problem shows up when the system later approaches the safe set boundary and the CBF constraint becomes active (indicating that a safety control action is needed). When this happens, $\eigval(x|\mathbb{D}_N) <0$ may not hold and there might not exist any control input direction along which the controllability of the CBF is dominant over the growth of the uncertainty. This would then compromise the recursive feasibility of the SOCP.}

The crux of our algorithm is to make sure we never end up in this situation. We accomplish this by applying control inputs (and adding those points to the dataset) \revision{that reduce the value of $\lambda_\dagger$} in a \textit{precautious} event-triggered fashion before the uncertainty grows in \revision{the safe control direction severely}.

Algorithm \ref{algo:safelearning} shows the concrete implementation of this framework. \revision{As a default, whenever the value of $\eigval$ lies under a threshold $-\thres <0$ (a negative constant close to $0$), the GP-CBF-SOCP \eqref{eq:gp-cbf-socp} determines a safety-filtered control. Under this case, the feasibility of the SOCP is guaranteed by Lemma \ref{lemma:sufficient}, so such a safe control input can always be determined.} However, if the value of $\eigval$ reaches $-\thres$, we take a control input along $\usafe$ and add the resulting measurement to the GP dataset, with the goal of reducing the uncertainty along the direction of $\usafe$ and consequently decreasing the value of $\eigval$ to below $-\thres$ for the following time steps. \revision{In addition to the event-triggered updates when $\eigval$ reaches $-\thres$, we also collect time-triggered measurements to update the dataset, with triggering period $\tau$, to constantly maintain low uncertainty.} Thus, Algorithm \ref{algo:safelearning} constructs a time-varying dataset $\mathbb{D}_{N(t)}$ and implicitly defines a closed-loop control law.

\setlength{\textfloatsep}{5pt}
\begin{algorithm}[t]
\small
\label{algo:safelearning}
\DontPrintSemicolon 
Initialize $t=0$, $x(0) = x_0$. Get $N(0)$, $\mathbb{D}_{N(0)}$.

\While{$t < T_\mathrm{max}$} {
    {\revision{// determine control input $u$}}\;
    $x \gets x(t)$\;
	$\eigval \gets \mathrm{getLambdaDagger}\big(x,\mathbb{D}_{N(t)}\big)$\;
    \eIf{$\eigval<-\thres$}
   {
   	 $u \gets u^*(x)$ from the SOCP \eqref{eq:gp-cbf-socp}\;
   	 
   }
   {$u \gets \usafe(x)$ from \eqref{eq:safe-direction}\;
   }
   {\revision{// experiment proceed for sampling time $\Delta t$}}\;
   {($t\gets t + \Delta t$)}\;
   {\revision{// update the dataset $\mathbb{D}_{N(t)}$}}\;
    \If{$(\eigval\geq-\thres)$ \textup{\textbf{or}} $(t \bmod \tau=0)$}
   {
   {\revision{Evaluate $z_B\!=\!\frac{x(t+\Delta t) - x(t)}{\Delta t}$ (Measurement of $\Delta_B(x,u)$)}}\;
   $\mathbb{D}_{N(t)} \gets \mathbb{D}_{N(t)} \cup \{(x,u),z_B\}$\;
   $N(t) \gets N(t)+1$\;}
   
}
\caption{Safe Online Learning}
\end{algorithm}

\begin{remark} \revision{(Computational complexity)} Using Algorithm \ref{algo:safelearning}, the number of data points $N$ grows with time. Since each data point is added individually, rank-one updates to the kernel matrix inverse can be computed in $O(N^2)$. However, if $N$ becomes large enough to compromise real-time computation, a smart-forgetting strategy such as the one of \cite{umlauft2020smart} could be applied. Note that there exists a wide variety of methods in the Sparse GP literature whose objective is to speed up the vanilla GP inference \cite{liu2020gaussian}. These methods can be considered complementary to our approach.
\end{remark}

\subsection{Theoretical Analysis}
\label{subsec:main-theory-part}
In this section, we provide theoretical results about the effectiveness of Algorithm \ref{algo:safelearning} in guaranteeing the safety of the unknown system \eqref{eq:system} with respect to the safe set $\mathcal{X}_{\text{safe}}$. We start by showing that with Algorithm \ref{algo:safelearning}, we can keep $\lambda_\dagger<0$ for the full trajectory under some assumptions.

\begin{assumption}
\label{assumption:initial_lambda}
\revision{We assume that we have an initial dataset $\mathbb{D}_{N(0)}$ and/or a nominal model \eqref{eq:nominal-model} such that at the initial state $x_0 \in \mathcal{X}$ and initial time $t=0$, we have $\eigval(x_0|\mathbb{D}_{N(0)})<0$.}
\end{assumption}

\revision{\noindent Note that Assumption \ref{assumption:initial_lambda} implies that we can find a safety direction $\usafe$ to begin with at the initial state (from Equation \ref{eq:safe-direction}). This assumption can be met using either a nominal model} \revision{\eqref{eq:nominal-model} or an initial small set of data $\mathbb{D}_{N(0)}$ that provides low uncertainty locally at the initial state. Note that this low uncertainty is only required to initiate the controller with a safe control action, and after the initialization of the controller, our algorithm will take care of preventing the infeasibility whenever the prior information from the nominal model and the existing data is deficient.}

\begin{assumption}
\label{assumption:reldegree1}
We assume that the CBF $B$ satisfies the relative degree one condition in $\mathcal{X}$, i.e., $L_gB(x) \neq 0\ \forall x \in \mathcal{X}$. Furthermore, we assume that for any $ x_0 \in \mathcal{X}$, for the trajectory $x(t)$ generated by running Algorithm \ref{algo:safelearning}, with $\mathbb{D}_{N(t)}$ being the dataset at time $t$, we have $\wfeas(x(t)|\mathbb{D}_{N(t)}) \neq 0\ \forall t$.
\end{assumption}

\revision{\noindent Note that the relative-degree-one condition in Assumption \ref{assumption:reldegree1} was already required to guarantee Lipschitz continuity of the solutions of the original CBF-QP \eqref{eq:cbf-qp-all}, as explained in Section \ref{sec:02background}. The second part of the assumption is stating that the mean prediction from the GP model for $L_gB(x)$ must also satisfy this condition.}

\begin{assumption}
\label{assumption:triggering_edagger}
Running Algorithm \ref{algo:safelearning} from any $ x_0\!\in\!\mathcal{X}$, let $\{t_\kappa\}_{\kappa \in \mathbb{N}}$ be the sequence of times at which
$\eigval(x(t)|\mathbb{D}_{N(t)})\geq \!-\thres$. We assume that for every $ \kappa$ we have $\wfeas\big(x(t_\kappa)|\mathbb{D}_{N(t_\kappa)+1}\big) \eigvec\big(x(t_\kappa)|\mathbb{D}_{N(t_\kappa)}\big) \neq 0$, where $\mathbb{D}_{N(t_\kappa)+1}$ is the resulting dataset after the event-triggered update. 
\end{assumption}

Assumption \ref{assumption:triggering_edagger} makes sure that during an event-triggered update of the dataset, the new mean controllability direction of the CBF \revision{is not completely orthogonal to the previous safe direction. A violation of either Assumption \ref{assumption:reldegree1} or \ref{assumption:triggering_edagger} would require the direction of $\wfeas(x)$ to drastically change in one timestep, which is unlikely to occur under Lemma \ref{lemma:DeltaUCB}.}

\revision{
With the above assumptions, we can now show that the eigenvalue $\lambda_\dagger$ will be maintained negative under Algorithm \ref{algo:safelearning} for the entire interval of the existence of the system solution:
}

\begin{lemma}
\label{lemma:algo_lambda}
Under Assumptions \ref{assumption:initial_lambda}, \ref{assumption:reldegree1} and \ref{assumption:triggering_edagger}, for all $ x_0 \in \mathcal{X}$, let $x(t)$ be the trajectory generated by running Algorithm \ref{algo:safelearning} for system \eqref{eq:system}. Let $\mathbb{D}_{N(t)}$ be the time-varying dataset generated during the execution of Algorithm \ref{algo:safelearning}. If the trajectory $x(t)$ exists and is unique during some time interval $t \in [0, \tau_{max})$, then it holds that $\eigval\big(x(t)|\mathbb{D}_{N(t)}\big)<0$ for all $t \in [0, \tau_{max})$.
\end{lemma}
\begin{proof}
See Appendix \ref{subsec:proof-lemma-lambda}.
\end{proof}

Previously, Lemma \ref{lemma:sufficient} demonstrated that $\eigval(x | \mathbb{D}_N) < 0$ is a sufficient condition for pointwise feasibility of the GP-CBF-SOCP at a point $x \in \mathcal{X}$ using a dataset $\mathbb{D}_N$. Now, Lemma \ref{lemma:algo_lambda} ensures that $\eigval(x(t) | \mathbb{D}_{N(t)}) < 0$ always holds along each trajectory $x(t)$ and dataset $\mathbb{D}_{N(t)}$ obtained by running the safe learning algorithm. Therefore, we can establish recursive feasibility of the GP-CBF-SOCP when using the proposed safe learning strategy, as formalized in the following statement.

\begin{theorem}[Recursive feasibility of the GP-CBF-SOCP]
\label{thm:recursive_feas}
Under Assumptions \ref{assumption:initial_lambda}, \ref{assumption:reldegree1} and \ref{assumption:triggering_edagger}, for all $ x_0 \in \mathcal{X}$ let $x(t)$ be the trajectory generated by running Algorithm \ref{algo:safelearning} for system \eqref{eq:system}. Let $\mathbb{D}_{N(t)}$ be the time-varying dataset generated during the execution of Algorithm \ref{algo:safelearning}. If the trajectory $x(t)$ exists and is unique during some time interval $t \in [0, \tau_{max})$, then the probabilistic safety constraint \eqref{eq:socp-cbf-constraint} is feasible at all times $t \in [0, \tau_{max})$ for the trajectory $x(t)$ and dataset $\mathbb{D}_{N(t)}$.
\end{theorem}
\begin{proof}
This is a direct consequence of Lemmas \ref{lemma:sufficient} and \ref{lemma:algo_lambda}.
\end{proof}

As a next step, we wish to remove the existence and uniqueness assumption of Theorem \ref{thm:recursive_feas}. We do this by proving that the trajectory $x(t)$ generated by running Algorithm \ref{algo:safelearning} in fact does locally exist and is unique.
Note that the policy that Algorithm \ref{algo:safelearning} defines is a switched control law, since new data points are added at discrete time instances. We start by showing that for a fixed dataset $\mathbb{D}_N$, the solution of the GP-CBF-SOCP is locally Lipschitz continuous under some assumptions:

\begin{assumption}
\label{assumption:contdiff}
We assume that the Lie derivatives of $B$ computed using the nominal model $L_{\tilde{f}}B(x)$, $L_{\tilde{g}}B(x)$, as well as the function $\gamma(B(x))$, the reference policy $u_\text{ref}(x)$ and the GP prediction functions for any fixed dataset $\gpmuB(x,u),\ \gpsigmaB(x,u)$  are twice continuously differentiable in $x$, $\forall x \in \mathcal{X}.$
\end{assumption}
Note that the GP prediction functions are twice continuously differentiable in $x$ when the components $k_1,\ldots, k_{m+1}$ of the ADP compound kernel \eqref{eq:adpkernel} use the squared exponential kernel or other common kernels that are twice continuously differentiable.

\begin{lemma}[Lipschitz continuity of solutions of the GP-CBF-SOCP]
\label{lemma:lipsch_SOCP}
Under Assumption \ref{assumption:contdiff}, for a point $x \in \mathcal{X}$ and \revision{a fixed dataset $\mathbb{D}_N$} such that $\eigval(x|\mathbb{D}_N)<0$ holds, the solution of the GP-CBF-SOCP \eqref{eq:gp-cbf-socp} is locally Lipschitz continuous around $x$.
\end{lemma}
\begin{proof}
See Appendix \ref{subsec:proof-lemma-lipschitz}.
\end{proof}

\begin{remark}
To the best of our knowledge, Lemma \ref{lemma:lipsch_SOCP} is the first result concerning Lipschitz continuity of SOCP-based controllers using CBFs or, equivalently, Control Lyapunov Functions (CLFs) for a general control input dimension. Very recently, several SOCP-based frameworks have been developed for robust data-driven safety-critical control using CBFs and CLFs \cite{dhiman2021control, taylor2020towards, dean2021guaranteeing, buch2021robust, greeff2021learning}, and verifying the local Lipschitz continuity of the SOCP solution serves to guarantee local existence and uniqueness of trajectories of the closed-loop dynamics.
\end{remark}

Even though, from Lemma \ref{lemma:lipsch_SOCP}, for a fixed dataset the solution of the GP-CBF-SOCP is Lipschitz continuous, the control law defined by Algorithm \ref{algo:safelearning} can potentially be discontinuous due to the dataset updates. This fact makes the closed-loop system potentially non-Lipschitz when using Algorithm \ref{algo:safelearning}.

Using Lemmas \ref{lemma:algo_lambda} and \ref{lemma:lipsch_SOCP}, Theorem \ref{thm:existence_uniqueness} below establishes local existence and uniqueness of the solution of the closed-loop system under the switched control law defined by Algorithm \ref{algo:safelearning} (even when the closed-loop system is non-Lipschitz).

\begin{theorem}[Local existence and uniqueness of executions of the safe learning algorithm]
\label{thm:existence_uniqueness}
Under Assumptions \ref{assumption:initial_lambda}, \ref{assumption:reldegree1}, \ref{assumption:triggering_edagger} and \ref{assumption:contdiff},
there exists a $\tau_{max}>0$ such that for any $x_0 \in \mathcal{X}$ a unique solution $x(t)$ of \eqref{eq:system} under the control law defined by Algorithm \ref{algo:safelearning} exists for all $ t \in [0,\tau_{max})$.
\end{theorem}
\begin{proof}
See Appendix \ref{subsec:proof-theorem-existence}.
\end{proof}

Previously, Theorem \ref{thm:recursive_feas} gave conditions under which the probabilistic constraint \eqref{eq:socp-cbf-constraint} is recursively feasible when using the control law defined by Algorithm \ref{algo:safelearning}. This means that, with high probability, the true CBF constraint \eqref{eq:cbf-constraint} can be satisfied at every timestep, as follows from Lemma \ref{lemma:DeltaUCB}. This fact can now be combined with the local existence and uniqueness result of Theorem \ref{thm:existence_uniqueness} to establish forward-invariance of the safe-set $\mathcal{X}_{\text{safe}}$ with high probability, as was originally formulated in Problem \ref{prob:problem-statement}.

\begin{theorem}[Main result: forward invariance with high probability]
\label{thm:main_theorem}
Under Assumptions \ref{assumption:valid-cbf}, \ref{assumption:initial_lambda}, \ref{assumption:reldegree1}, \ref{assumption:triggering_edagger} and \ref{assumption:contdiff}, the control law defined by Algorithm \ref{algo:safelearning} applied to the true plant \eqref{eq:system} renders the set $\mathcal{X}_{\text{safe}} = \{ x \in \mathcal{X}: B(x) \geq 0 \}$ forward invariant with a probability of at least $1-\delta$.
\end{theorem}
\begin{proof}
Let the control law defined by Algorithm \ref{algo:safelearning} be denoted as $\ualgo(x)$. For all $ x_0 \in \mathcal{X}$, the solution $x(t)$ of \eqref{eq:system} under $\ualgo(x)$ satisfies $\eigval\big(x(t)|\mathbb{D}_{N(t)}\big)<0$, $\forall t \in [0,\tau_{max})$ with $\tau_{max}>0$ from Theorem \ref{thm:existence_uniqueness} and Lemma \ref{lemma:algo_lambda}. Here, $\mathbb{D}_{N(t)}$ is the time-varying dataset generated by Algorithm \ref{algo:safelearning}. Moreover, from Lemma \ref{lemma:sufficient} and Theorem \ref{thm:recursive_feas}, this means that the GP-CBF-SOCP is feasible $\forall t \in [0,\tau_{max})$. Furthermore, note that $\ualgo$ is the solution of \eqref{eq:gp-cbf-socp}, except at times when $\eigval\big(x(t)|\mathbb{D}_{N(t)}\big) \geq -\thres$ in which case it takes the value of $\usafe (x(t))$. However, since even at those times $\eigval\big(x(t)|\mathbb{D}_{N(t)}\big)<0$, $\usafe(x(t))$ is also a feasible solution of \eqref{eq:gp-cbf-socp}. Therefore, under $\ualgo$, the constraint \eqref{eq:socp-cbf-constraint} is satisfied for all $ t \in [0,\tau_{max})$. This fact, together with the probabilistic bound on the true plant CBF derivative $\dot{B}$ \eqref{eq:B_UCB} that arises from Lemma \ref{lemma:DeltaUCB}, leads to:
\vspace{-5pt}
\begin{multline}
\label{eq:prob-cbf-true}
\mathbb{P}\big\{\ \dot{B}\big(x(t),\ualgo(x(t))\big) + \gamma\big(B(x(t))\big) \geq 0,\\ \forall x_0 \in \mathcal{X},\ \forall t\in[0,\tau_{max})\ \big\} \geq 1-\delta.
\vspace{-9pt}
\end{multline}
Noting that the trajectory $x(t)$ is a continuous function of time that exists and is unique for all $ t\in[0,\tau_{max})$
(from Theorem \ref{thm:existence_uniqueness}), we can now use Assumption \ref{assumption:valid-cbf} and the bound of \eqref{eq:prob-cbf-true} to obtain
\vspace{-10pt}
\begin{multline}
\mathbb{P}\{\ \forall x_0 \in \mathcal{X}_{\text{safe}},\ x(0) = x_0 \implies \\ x(t) \in \mathcal{X}_{\text{safe}},\ \forall t \in [0,\tau_{max})\ \}\geq 1-\delta.
\end{multline}
This is precisely the expression that appears in Problem \ref{prob:problem-statement}, and it means 
that the trajectories $x(t)$ will not leave the set $\mathcal{X}_{\text{safe}} = \{x \in \mathcal{X}: B(x)\geq 0 \}$ for all $x_0 \in \mathcal{X}_{\text{safe}}$ with a probability of at least $1-\delta$, completing the proof.
\end{proof}

\begin{remark}
\label{rmk:invariance_vs_pointwise}
Theorem \ref{thm:main_theorem} establishes the forward invariance of $\mathcal{X}_{\text{safe}}$ with a probability of at least $1 - \delta$. This is possible because of the fact that Lemma \ref{lemma:DeltaUCB} is not a pointwise result on the deviation of the GP prediction at a particular point, but instead a probability bound on the combination of all of the possible deviations (for all $N$, $x$ and $u$). The note \cite{lew2021problem} provides an insightful discussion of this topic.
\end{remark}

\begin{remark}
\label{remark:clf}
Note that the proposed framework can be easily extended to the problem of safe stabilization by adding a relaxed probabilistic CLF constraint to the SOCP \eqref{eq:gp-cbf-socp}, as done in \cite{castaneda2021pointwise}. The entire theoretical analysis about feasibility and safety would directly follow as long as Assumption \ref{assumption:contdiff} is adapted to include the CLF-related terms.
\end{remark}

\begin{remark}
\label{remark:epsilon}
\revision{
(Choice of the threshold $\varepsilon$ for $\eigval$) Our analysis assumes that the exploration action of \eqref{eq:safe-direction} is taken instantaneously; however, in practical implementation, the action must be held for a sampling time, $\Delta t$, as in Algorithm \ref{algo:safelearning}. Thus, $\varepsilon$ is required to make sure that the $\lambda_\dagger$ does not exceed zero during this sampling time. Longer sampling time would require a larger threshold value $\varepsilon$; a detailed analysis of the relationship between the two values under the sampled data setup remains a future work direction.
} 
\end{remark}

\section{Simulation Results}
\label{sec:06examples}

\begin{figure}
\centering
\includegraphics[width=\columnwidth]{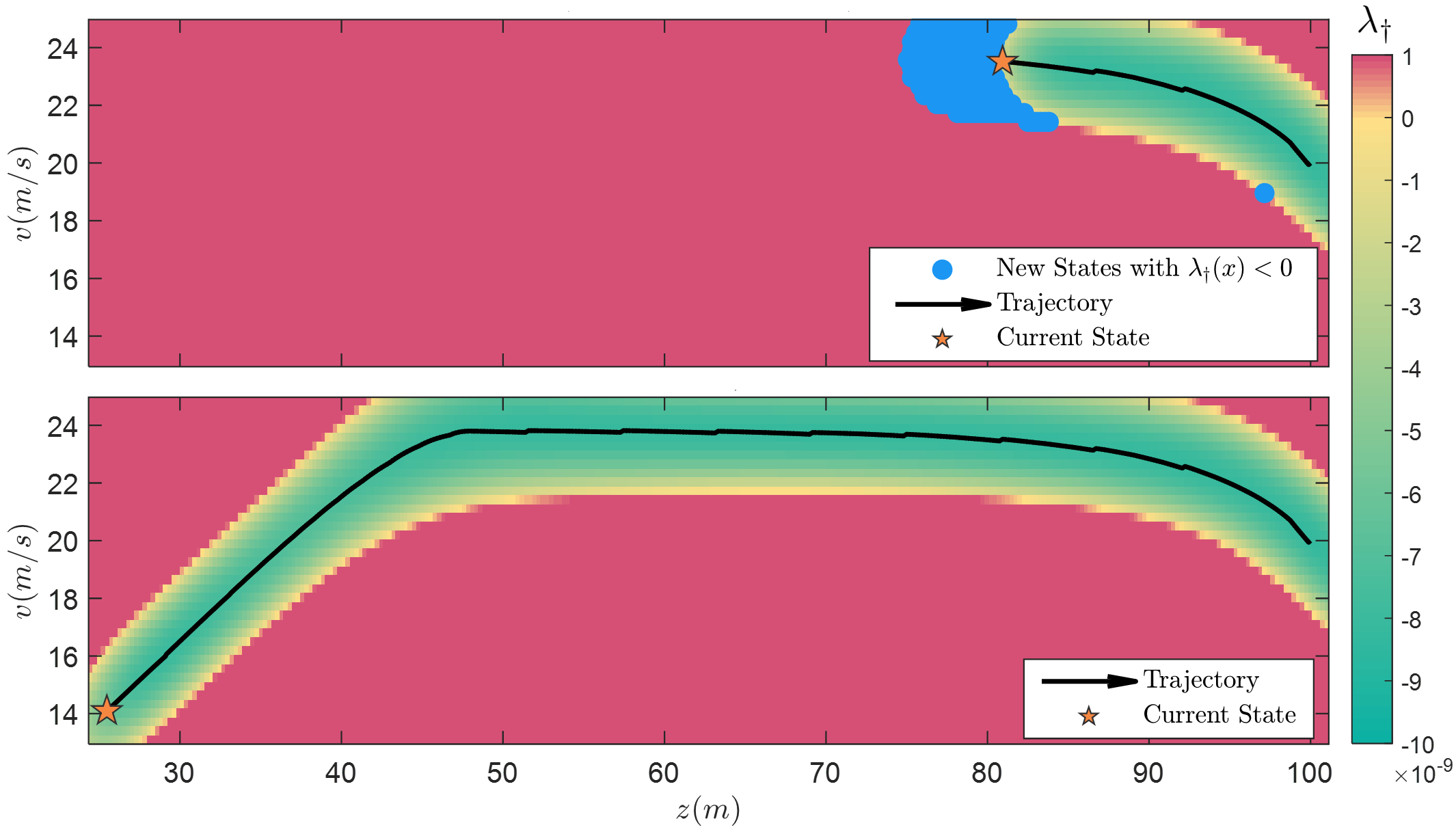}
\vspace{-7mm}
\caption{Color map of $\eigval$ in the state-space of the adaptive cruise control system $x\!=\![v, z]^T$ when running Algorithm \ref{algo:safelearning} with no prior data. The region in which $\eigval<0$ is expanded online as Algorithm \ref{algo:safelearning} collects new measurements. Top: snapshot when $\eigval$ hits the threshold $-\thres$, Algorithm \ref{algo:safelearning} collects a measurement along $\usafe$ which expands the region where $\eigval<0$ (in blue). Bottom: result at the end of the trajectory.}
\label{fig:lambda_color}
\end{figure}

\begin{figure}
\centering
\includegraphics[width=0.9\columnwidth]{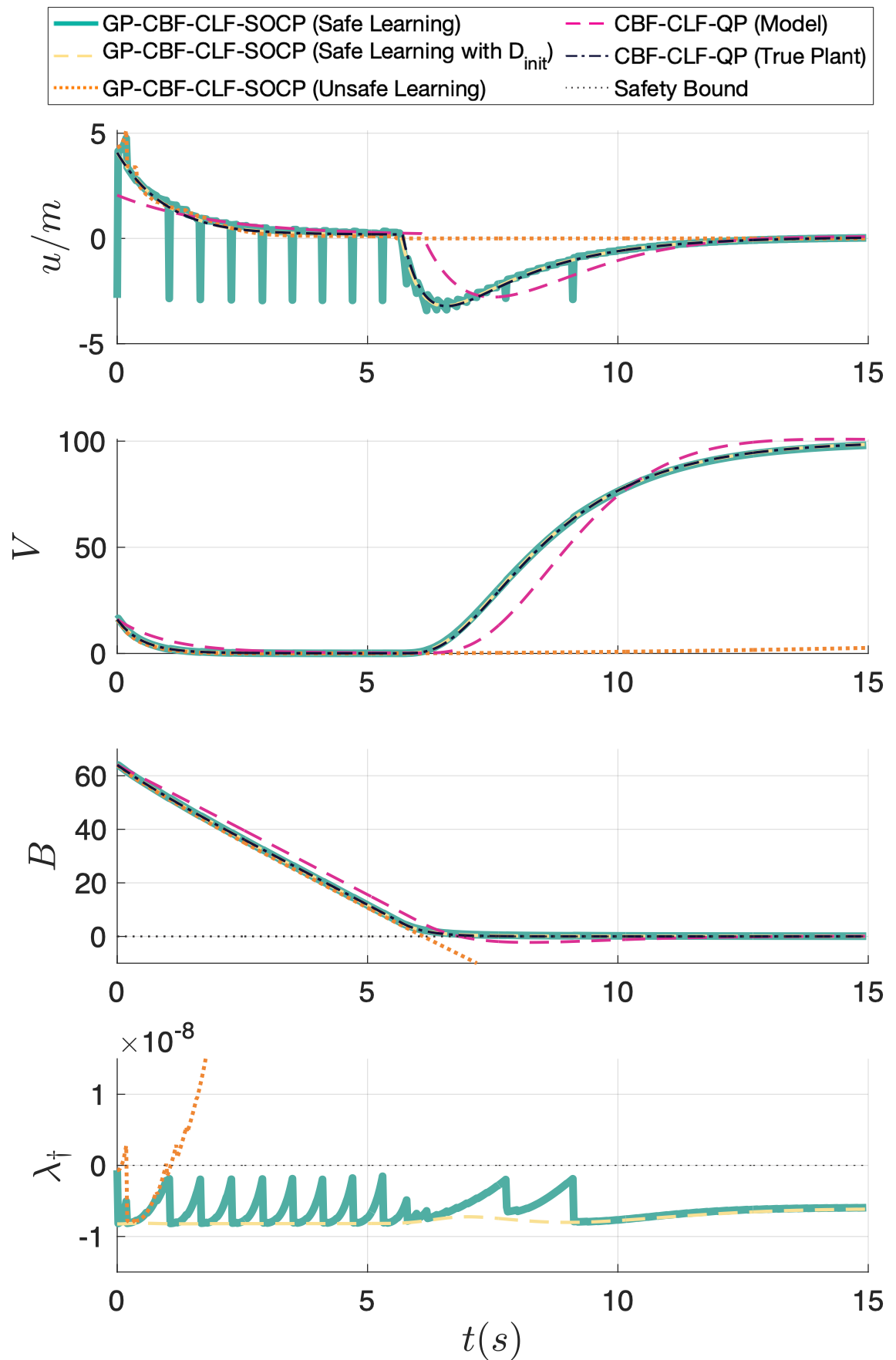}
\vspace{-3mm}
\caption{Simulation results of an adaptive cruise control system under model uncertainty, when controlled using different strategies: Algorithm \ref{algo:safelearning} with no prior data (green); Algorithm \ref{algo:safelearning} with a prior dataset (yellow); the GP-CBF-SOCP with no prior data using time-triggered updates online (orange); the CBF-QP using the uncertain dynamics (pink); and the oracle true-plant-based CBF-QP (black). Even when no prior data is available, Algorithm \ref{algo:safelearning} keeps the system safe ($B>0$) by collecting measurements in the safety direction (negative $u$) when $\eigval$ approaches $0$. This results in the ego car checking the brakes to reduce the uncertainty (negative spikes in the top plot). Using either the GP-CBF-SOCP with just time-triggered data collection, or the nominal model-based CBF-QP, the system becomes unsafe, as shown in the $B$ plot.
}
\label{fig:results}
\end{figure}

\textblack{In this section, we test our framework on the following two examples in numerical simulation. The first example of an adaptive cruise control system highlights how the feasibility of the controller improves from the data collected online through Algorithm \ref{algo:safelearning}. The second example of a kinematic vehicle system demonstrates the applicability of our framework to multi-input systems.}
\revision{In both examples, we demonstrate the recursive feasibility of our proposed controller by illustrating that $\lambda_\dagger < 0$ always holds during simulation time.}

\subsection{Adaptive Cruise Control}
\textblack{We apply our proposed framework to a numerical model of an adaptive cruise control system \vspace{-1em}

{\small

\begin{equation}
    \dot{x} = f(x) + g(x)u, ~~f(x)\! =\! \begin{bmatrix} - F_r(v)/m \\ v_0 - v \end{bmatrix},\
    g(x)\! =\! \begin{bmatrix} 0 \\ 1/m \end{bmatrix},
\end{equation}}

\noindent where $x\!=\![v, z]^T\in \R^2$ is the system state, with $v$ being the ego car's velocity and $z$ the distance between the ego car and the car in front of it; $u \in \R$ is the ego car's wheel force; $v_0$ is the constant velocity of the front car (14 m/s); $m$ is the mass of the ego car; and $F_r(v)\!=\!f_0\!+\!f_1 v\!+\!f_2 v^2$ is the rolling resistance force on the ego car. We introduce uncertainty in the mass and the rolling resistance.

A CLF is designed for stabilizing to a desired speed of $v_d = 24$ m/s, and a CBF enforces a safe distance of $z \geq 1.8 v$ with respect to the front vehicle. We specifically use $V(x) = (v-v_d)^2$ and $B(x)\!=\!z\!-\!1.8 v$.} \revision{Despite the uncertainty, the velocity term can still be regulated by the car wheel force (thus satisfying our assumption that $L_g B \neq 0$), ensuring that these} \revision{functions are a valid CLF and CBF.} 
Following Remark \ref{remark:clf}, the CLF-based stability constraint is added as a soft constraint to the SOCP controller, replacing the reference control input $u_{\text{ref}}$. Therefore, the CBF acts as a hard safety constraint that filters a control policy based on the CLF whose objective is to stabilize the car to the desired speed $v_d$.

Figure \ref{fig:lambda_color} shows a state-space color map of the value of $\eigval$ at two different stages of the trajectory generated running Algorithm \ref{algo:safelearning} for the adaptive cruise control system with no prior data starting from $x_0 = [20,100]^T$. The top plot represents an intermediate state, in which the system is still trying to reach the desired speed of $24$ m/s since the safety constraint \eqref{eq:socp-cbf-constraint} is not active yet (the car in front is still far). Even though Algorithm \ref{algo:safelearning} is collecting measurements in a time-triggered fashion using the SOCP \eqref{eq:gp-cbf-socp} controller, the state gets close to the boundary of $\eigval = 0$ frequently, since the performance-driven control input obtained from the SOCP \eqref{eq:gp-cbf-socp} when the safety constraint is not active is very different from $\usafe$. One such case is visualized in the top figure. However, Algorithm \ref{algo:safelearning} detects that $\eigval$ is getting close to zero and an event-triggered measurement in the direction of $\usafe$ is taken, which expands the region where $\eigval < 0$. The bottom plot shows the color map of $\eigval$ at the end of the process, with the final dataset. The safety constraint was active for a portion of the trajectory (when the ego vehicle approached the front one and reduced its speed), and the system stayed safe by virtue of using Algorithm \ref{algo:safelearning} to keep a direction $\usafe$ available.

Figure \ref{fig:results} shows that while a nominal CBF-QP (in pink) fails to keep the system safe under model uncertainty, Algorithm \ref{algo:safelearning} with no prior data (in green) always manages to keep $B>0$ and $\eigval < 0$ by collecting measurements along $\usafe$ (negative spikes in the control input $u$ in the top plot) when triggered by the event $\eigval\geq-\thres$. The same algorithm without these event-triggered measurements fails (orange), since when the safety constraint \eqref{eq:socp-cbf-constraint} becomes active, $\eigval$ soon gets positive and the SOCP \eqref{eq:gp-cbf-socp} becomes infeasible.

From another perspective, Figure \ref{fig:results} shows the importance of having a good nominal model or a prior database that properly characterizes a safe control direction. As shown in yellow, with such prior information Algorithm \ref{algo:safelearning} keeps the system safe without having to take any measurements along $\usafe$. If no prior data is given, the control law is purely learned online, which leads to $\eigval$ getting close to zero several times in the trajectory, and steps in the direction of $\usafe$ (negative $u$) are needed in order to prevent $\eigval$ from actually reaching zero. This clearly damages the desired performance, as the car would be braking from time to time. Nevertheless, this is required in order to be certain about how the system reacts to pressing the brake. Therefore, the proposed event-triggered design allows Algorithm \ref{algo:safelearning} to automatically reason about whether the available information is enough to preserve safety or the collection of a new data point along $\usafe$ is required instead. 
Note that our algorithm is also useful for cases in which a large dataset is available a priori, since it would secure safety even when the system is brought to out-of-distribution regions.

\begin{figure}
\centering
\includegraphics[width=0.9\columnwidth]{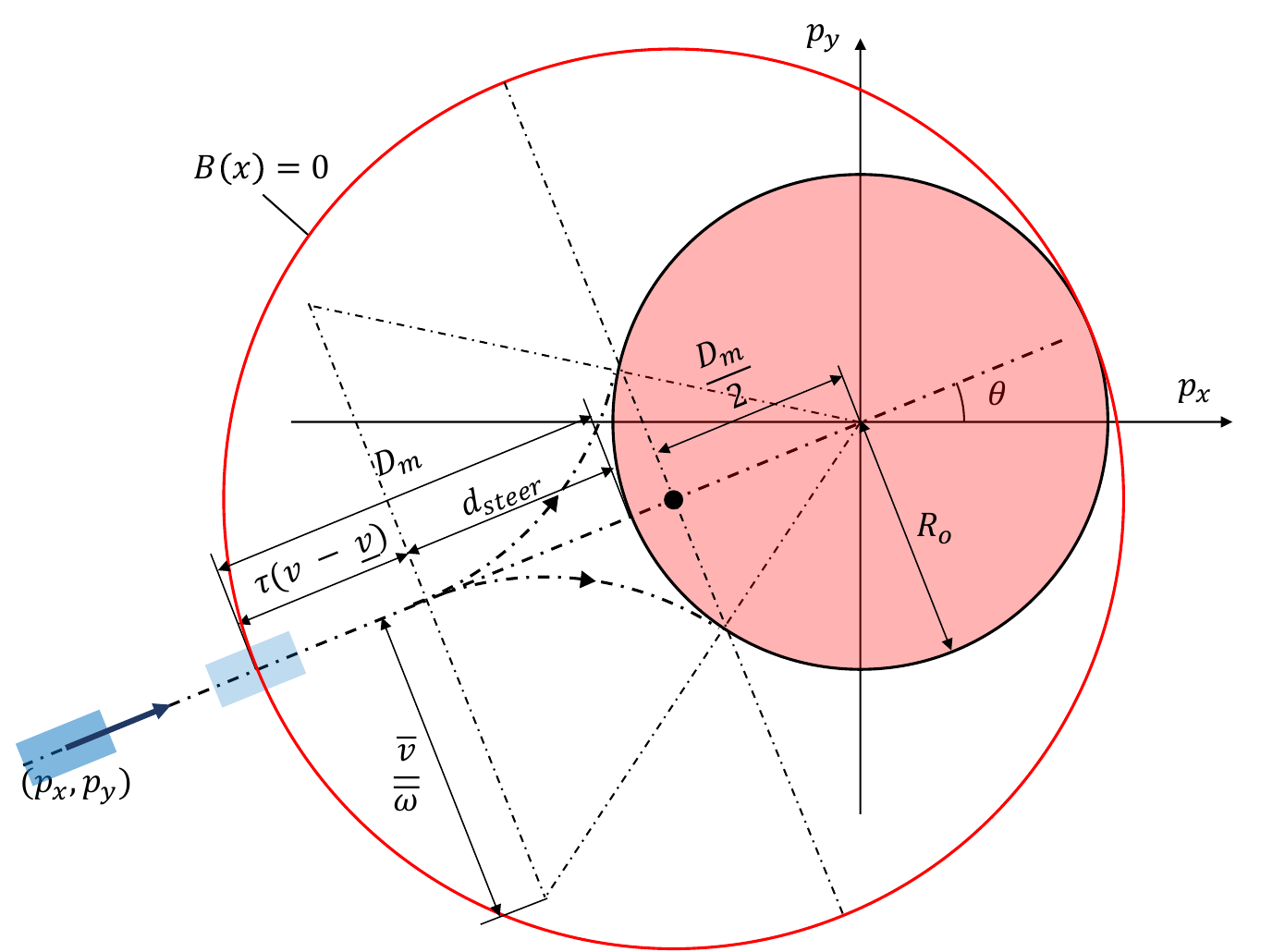}
\vspace{-0.5em}
\caption{Illustration of the zero-level set of the CBF for the kinematic vehicle example. $D_m$ is the safety distance, which is computed by adding the minimum distance for the vehicle to steer with a maximal yaw rate without colliding with the obstacle $d_{steer}$ and a velocity-dependent distance margin $\tau (v - \underbar{$v$})$.}
\label{fig:car4d_cbf}
\end{figure}

\subsection{Kinematic Vehicle}

\begin{figure}
\centering
\includegraphics[width=\columnwidth]{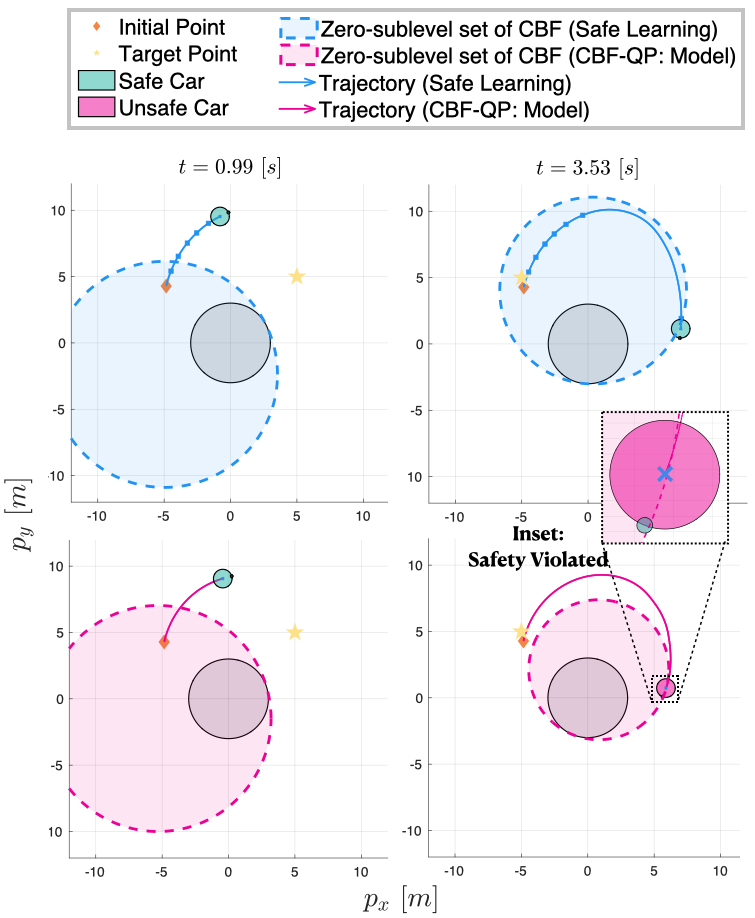}
\vspace{-1em}
\caption{Snapshots that show the chronological evolution of a 4-dimensional kinematic vehicle system under model uncertainty, when controlled using different methods: Algorithm \ref{algo:safelearning} with no prior data (top row, blue); the CBF-QP based on the nominal model (bottom row, pink). Starting at the initial state $x_0$ (orange diamond), the vehicle pursues the target (yellow star), while not colliding with the obstacle (grey circle). The curved line indicates the trajectory of the vehicle's position that terminates with its position at the time when a snapshot is taken (green or red circle). Note that the circle is colored red when the vehicle violates the safety constraint (i.e., $B(x) < 0$). The blue square positioned along the trajectory highlights the time stamps at which Algorithm \ref{algo:safelearning} collects the data in event-triggered manner. Finally, the filled circle with a dotted border represents the zero-sublevel set of CBF. To watch the full video of the vehicle running under each control algorithm, please visit \href{https://www.youtube.com/watch?v=HM_VB_mGgeA}{https://youtu.be/HM\_VB\_mGgeA}.}
\label{fig:car4d_snapshots}
\end{figure}

\begin{figure}
\centering
\includegraphics[width=\columnwidth]{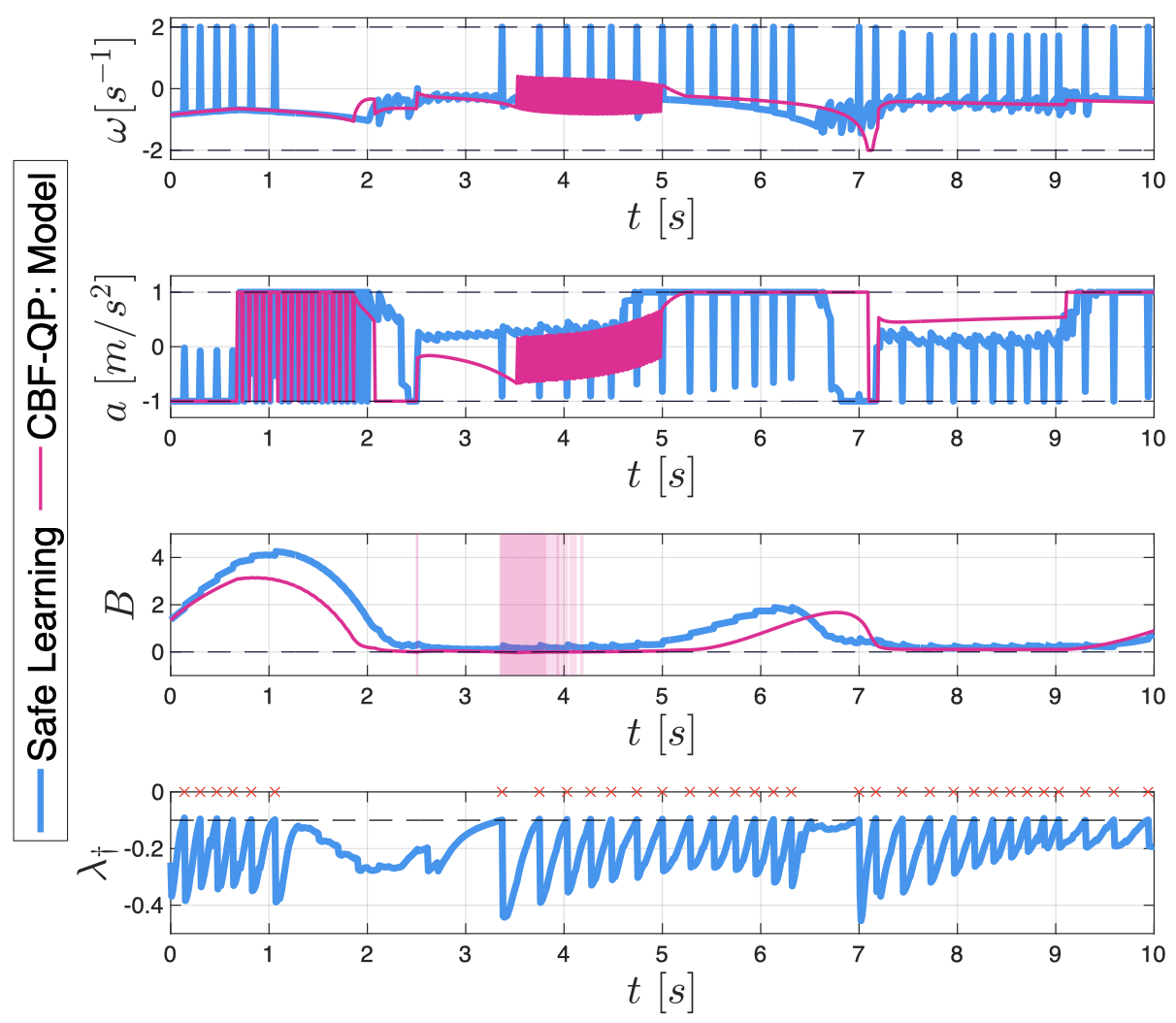}
\caption{Simulation results of 4-dimensional kinematic vehicle system under model uncertainty, when using two strategies introduced in Figure \ref{fig:car4d_snapshots} with the identical color notation. The four plots illustrate the yaw rate, the acceleration control inputs, the CBF values, and $\lambda_\dagger$ in time respectively. The dotted lines denote the input bounds, the zero-level of the CBF $B(x)=0$; and the threshold $-\epsilon$ in Algorithm \ref{algo:safelearning}. The red bars in the third plot represent the time stamps when the nominal CBF-QP violates safety. In contrast, Algorithm \ref{algo:safelearning} ensures $B(x) > 0$ at all times. The red cross points in the last plot indicate the time stamps when $\lambda_\dagger$ hits $-\epsilon$ and the safe exploration is executed according to Algorithm \ref{algo:safelearning}.}
\label{fig:car4d_results}
\end{figure}

\textblack{Next, in order to gauge our framework's applicability to systems with higher state dimensions and multiple control inputs, we apply our method to a four-dimensional kinematic vehicle system. The state vector is denoted as $x\!=\![p_x, p_y, \theta, v]^T\in \R^4$ which consists of the vehicle's position $(p_x, p_y)$, heading angle $\theta$, and longitudinal velocity $v$; the control input is denoted as $u\!=\![w,a]^T \in \R^2$ which includes the vehicle's yaw rate $w \in [-\bar{w}, \bar{w}]$ and the longitudinal acceleration $a \in [-\bar{a}, \bar{a}]$. We use the values $\bar a = 1$ and $\bar w = 2$. The dynamics of the system are modeled as
\begin{equation}
\label{eq:car4d-dynamics}
\small
    f(x)\! =\! \begin{bmatrix} k_v v \cos\theta \\ 
                               k_v v \sin\theta \\ 
                               0 \\ 
                               -\mu v+ s_e h(p_x,p_y) 
               \end{bmatrix},\
    g(x)\! =\! \begin{bmatrix} 0   & 0\\
                               0   & 0\\
                               k_w & 0\\
                               0   & k_a \end{bmatrix},
\normalsize
\end{equation}
\noindent where $k_v$, $k_w$, $k_a$ are coefficients that capture the skid, the term $\mu v$ represents the drag, and $s_e h(p_x, p_y)$ accounts for the effect of the slope of the terrain. We assume that the nominal model does not address such effects (i.e., $k_v=k_w=k_a=1,\ \mu=s_e=0$), while the uncertainty imposed on the true} system is induced by $k_v=2$, $k_w=1.5$, $k_a=1$, $\mu=0.5$, $s_e=0.5$, $h(p_x, p_y)=(p_x^2+p_y^2)^{0.1}$. Note that unstructured uncertainties are imposed through terms like $h(p_x, p_y)$, which can be arbitrary functions, unlike the previous example that only imposes parametric uncertainties.

As illustrated in Figures \ref{fig:car4d_cbf} and \ref{fig:car4d_snapshots}, the objective of the control is to reach the target points alternating in time while not colliding into a static circular obstacle of radius $R_o=3$ centered at the origin. The reference controller $u_{\text{ref}}$ has two objectives: 1) it pursues a target point that alternates among a given set of points $S_T=\{(5,5), (5,-5), (-5,-5), (-5,5)\}$ every $p=2.5$ seconds, and 2) it stabilizes the vehicle's velocity to $v_{d}$ while assuring that it is always bounded in $[\underbar{$v$}, \bar{v}]$. We use the values $\underbar{$v$}=1$, $\bar{v}=5$, and $v_d=3$. All units are in the metric system.

The CBF we use is
\small
\begin{equation*}
    B(x)\! =\! \sqrt{\left(p_x \!+\! \frac{D_m}{2} \cos \theta\right)^2 \!+ \!\left(p_y\!+\!\frac{D_m}{2} \sin \theta\right)^2}- \left(\!R_o\!+\!\frac{D_m}{2}\!\right),
\end{equation*}
\normalsize
\noindent where
\small
\begin{equation*}
D_m = \tau(v-\underbar{$v$}) + d_{steer}; \quad d_{steer} = R_o\sqrt{1+\frac{2\bar{v}}{R_o \bar{w}}}-R_o.
\end{equation*}
\normalsize

\noindent This CBF adds a safety margin $D_m$ to the obstacle in the direction of the vehicle's heading angle, based on its minimum velocity and maximum steering rate as shown in Figure \ref{fig:car4d_cbf}. We can analytically check that the zero-superlevel set of the CBF is control invariant and that the CBF constraint is always feasible under the input bounds. \revision{We can also check that the model uncertainty in this example does not affect $L_g B(x)$ to} \revision{vanish to zero, thereby allowing us to use it as a CBF for the true plant.}

Figures \ref{fig:car4d_snapshots} and \ref{fig:car4d_results} illustrate that while the CBF-QP based on the nominal model (in pink) escapes the zero-superlevel set of the CBF, Algorithm \ref{algo:safelearning} (in blue) without any prior data always keeps the vehicle 
inside.
This result not only demonstrates the validity of the proposed strategy when applied to a multi-input system but also alludes to the intuition behind our strategy: when $\lambda_{\dagger}$ hits $-\epsilon$, the vehicle steers away from the obstacle and decelerates more in order to improve the certainty of its safe control direction.

\section{Conclusion}
\label{sec:07conclusion}

In this article, we have introduced a Control Barrier Function-based approach for the safe control of uncertain systems. Our results show that it is possible to guarantee the invariance of a safe set for an unknown system with high probability, by combining any available approximate model knowledge and sufficient data collected from the real system. We achieve this by first introducing a safety-critical optimization-based controller that, by formulation, is probabilistically robust to the prediction uncertainty of the unknown system's dynamics. However, this optimization problem only produces a safe control action when the available information about the system (prior model knowledge and data) is sufficiently rich, as our feasibility analysis shows. As a means to fulfill this feasibility requirement, we later presented a formal method that, by collecting data online when required, is able to guarantee the recursive feasibility of the controller and therefore preserve the unknown system's safety with high probability.
Algorithm \ref{algo:safelearning} presents a simple embodiment of this idea;
however, we believe that future work should not be restricted to this particular implementation, since the most important contribution of this article is a principled
reasoning procedure for conducting safe exploration when using data-driven control schemes.

Finally, we would like to emphasize that, as explained in Section \ref{sec:06examples}, a practical takeaway from the results of this paper is that approximate model knowledge and prior data serve to reduce the conservatism of safety-assuring data-driven control approaches. We are convinced that designing control strategies with useful safety guarantees for uncertain systems requires combining model-based and data-driven methods, which until very recently were seen as mutually exclusive by experts in the field. With this research, we aim to present new evidence of the potential benefits of combining the two approaches.

\section*{Acknowledgements}
We would like to thank Andrew J. Taylor and Victor D. Dorobantu for the insightful discussions on this research topic. Furthermore, we would like to thank Brendon G. Anderson for his suggestions regarding Lemma \ref{lemma:lipsch_SOCP}.

\begin{appendix}[Proofs and Intermediate Results]
\label{appendix:proofs}

\subsection{An Equivalent Formulation of the Chance Constraint}
We first provide an additional reformulation of the CBF chance constraint, equivalent to those of \eqref{eq:socp-cbf-constraint} and \eqref{eq:constraint_std}, which will be useful for the proofs of the feasibility results.
\begin{lemma}
\label{lemma:fundamental}The CBF chance constraint \eqref{eq:socp-cbf-constraint} is feasible at a point $x \in \mathcal{X}$ if and only if there exists a control input $u\in\R^m$ that satisfies both of the following conditions:
\begin{subequations}\label{eq:feas_lemma1}
\begin{numcases}{}
    [1\ u^T] H(x|\mathbb{D}_N) \begin{bmatrix} 1 \\u \end{bmatrix} \leq 0,    \label{eq:feas_lemma1_1}\\
    \socc u + \socd \geq 0,     \label{eq:feas_lemma1_2}
\end{numcases}
\end{subequations}
where
\begin{equation}
\label{eq:H_def}
    H(x|\mathbb{D}_N):=\begin{bmatrix}  H_{11} & H_{1u} \\
    H_{1u}^T & H_{uu} \end{bmatrix},\ \  \text{with}\vspace{-0.5em}
\end{equation}
\begin{align*}
\small
H_{11} = &\beta^2\socbb - (\socd)^2,\\
\begin{split}
H_{1u} = &\beta^2\socb^T\socA - \\ & \big(\socd\big)\socc ,\end{split}\\
H_{uu} = &\beta^2\socAA - \socc^T\socc.
\normalsize
\end{align*}
\end{lemma}

\noindent In \eqref{eq:H_def}, we have used the following relations:
\begin{align*}
\socbb &= \socb^T\socb\ \in \R,\\
\socAA &= \socA^T\socA\ \in \R^{m\times m}.
\end{align*}

\begin{proof}
The first inequality \eqref{eq:feas_lemma1_1} is directly obtained by squaring both sides of \eqref{eq:constraint_std}. The second inequality \eqref{eq:feas_lemma1_2} is required to check that the right-hand side of \eqref{eq:constraint_std} is non-negative, as the left-hand side is trivially non-negative.
\end{proof}

\subsection{An Intermediate Result of a Necessary Condition for Pointwise Feasibility}
\begin{lemma}
\label{lemma:necessary_H}
For a given dataset $\mathbb{D}_N$, if the GP-CBF-SOCP \eqref{eq:gp-cbf-socp} is feasible at a point $x \in \R^n$, then the symmetric matrix $H(x|\mathbb{D}_N)$ defined in \eqref{eq:H_def} cannot be positive definite.
\end{lemma}
\begin{proof}
Positive definiteness of $H(x|\mathbb{D}_N)$ would mean that there does not exist any control input $u \in \R^m$ such that $[1\ u^T] H(x|\mathbb{D}_N) \begin{bmatrix} 1\\ u \end{bmatrix} \leq 0$.
However, this is a contradiction to Equation \eqref{eq:feas_lemma1_1} in Lemma \ref{lemma:fundamental}. Therefore, $H(x|\mathbb{D}_N)$ cannot be positive definite if the GP-CBF-SCOP is feasible.
\end{proof}

\subsection{Proof of Lemma \ref{lemma:necessary}}
\label{subsec:proof-lemma-necessary}
In this proof, we show that the condition \eqref{eq:necessary} of Lemma \ref{lemma:necessary} is equivalent to $H(x|\mathbb{D}_N)$ of \eqref{eq:H_def} not being positive definite for the same state $x \in \mathcal{X}$ and dataset $\mathbb{D}_N$. By Lemma \ref{lemma:necessary_H}, this would mean that \eqref{eq:necessary} is a necessary condition for pointwise feasibility of the GP-CBF-SOCP, which is the desired result.

Let $\psi(x|\mathbb{D}_N) \coloneqq [\socd,\ \socc]$. Then, condition \eqref{eq:necessary} does not hold if and only if
\begin{equation}
\label{eq:necessary_converse_prev}
1- \psi(x|\mathbb{D}_N) \frac{1}{\beta^2}\gpGramB(x|\mathbb{D}_N)^{-1}\psi(x|\mathbb{D}_N)^T >0.
\end{equation}
Note that \eqref{eq:necessary_converse_prev} is equivalent to 
\begin{equation}
\label{eq:necessary_converse}
M(x|\mathbb{D}_N)/(\beta^2 \gpGramB(x|\mathbb{D}_N)) > 0,
\end{equation}
where we use the operator $/$ for the Schur complement, and $M(x|\mathbb{D}_N) := \begin{bmatrix} 1 & \psi(x|\mathbb{D}_N)\\ \psi(x|\mathbb{D}_N)^T &  \beta^2 \gpGramB(x|\mathbb{D}_N)\end{bmatrix}$.
From \cite[Thm. 1.12]{zhang2006schur}, since $\gpGramB(x|\mathbb{D}_N)$ is positive definite, \eqref{eq:necessary_converse} holds if and only if $M(x|\mathbb{D}_N)$ is also positive definite.
We now apply again \cite[Thm. 1.12]{zhang2006schur}, but this time to $M(x|\mathbb{D}_N)/1$. Then, \eqref{eq:necessary_converse} is equivalent to the positive definiteness of $M(x|\mathbb{D}_N)/1 = \beta^2 \gpGramB(x|\mathbb{D}_N) - \psi(x|\mathbb{D}_N)^T\psi(x|\mathbb{D}_N)$\;$=\beta^2[\socb \ \socA]^T [\socb \ \socA]-\psi(x|\mathbb{D}_N)^T\psi(x|\mathbb{D}_N)\!=\!H(x|\mathbb{D}_N)$. Therefore, \eqref{eq:necessary} does not hold if and only if $H(x|\mathbb{D}_N)$ is positive definite, and the inverse statement completes the proof.
\qed

\subsection{Proof of Lemma \ref{lemma:sufficient}}
\label{subsec:proof-lemma-sufficient}
For a point $x\in \mathcal{X}$ and dataset $\mathbb{D}_N$, let $\eigvec(x|\mathbb{D}_N) \in \R^{m\times 1}$ be the unit eigenvector of $\feasmat \in \R^{m\times m}$ associated with the minimum eigenvalue $\eigval(x|\mathbb{D}_N) \in \R$. Then, clearly, 
\begin{equation}
\label{eq:eigval_eigvec_imply}
\eigval(x|\mathbb{D}_N)\!<\!0\!\implies\!\eigvec(x|\mathbb{D}_N)^T \feasmat \eigvec(x|\mathbb{D}_N)\!<\!0.
\end{equation}
Using \eqref{eq:eigval_eigvec_imply} and taking into account the definition of $\feasmat$ in \eqref{eq:defineF}, the fact that $\socAA$ is positive definite indicates that $\eigval(x|\mathbb{D}_N)\!<\!0\!\implies\!\socc \eigvec(x|\mathbb{D}_N) \neq 0$.

Next, take a control input $\usafe(x)$ in the direction of $\eigvec(x|\mathbb{D}_N)$, as defined in \eqref{eq:safe-direction}.
Plugging $\usafe(x)$ into \eqref{eq:feas_lemma1_1}, the left-hand side of \eqref{eq:feas_lemma1_1} becomes a polynomial in $\alpha$, of the form
$\alpha^2 \eigvec(x|\mathbb{D}_N)^T \feasmat \eigvec(x|\mathbb{D}_N) + \mathcal{O}(\alpha)$, where  $\mathcal{O}(\alpha)$ denotes terms with degree lower than or equal to $1$.
Note that the value of the polynomial can be made negative by choosing a large-enough constant $\alpha$, since from \eqref{eq:eigval_eigvec_imply} we know that $\eigvec(x|\mathbb{D}_N)^T \feasmat \eigvec(x|\mathbb{D}_N) < 0$. Lastly, also plugging $\usafe(x)$ into \eqref{eq:feas_lemma1_2} yields $\alpha|\socc \eigvec(x|\mathbb{D}_N)| + \socd \geq 0$, which again holds for a sufficiently large $\alpha$. Therefore, by Lemma \ref{lemma:fundamental} the GP-CBF-SOCP \eqref{eq:gp-cbf-socp} is feasible when $\eigval(x|\mathbb{D}_N) <0$ and $\usafe(x)$ is a feasible control input for large-enough $\alpha$.
\qed

\subsection{Proof of Theorem \ref{th:nec_and_suf}}
\label{subsec:proof-theorem-necessary-sufficient}
Since $\socA \!\succ\!0$, the GP-CBF-SOCP \eqref{eq:gp-cbf-socp} is feasible if and only if an intersection between the \textit{hyperboloid} $\beta\norm{\socA u\! +\! \socb}_2\! =\! t$ and the \textit{hyperplane} $\socc u\! +\! \socd\! =\! t$ exists. As $\norm{u}_2 \xrightarrow{} \infty$,
the hyperboloid asymptotically converges to the conical surface $\beta\norm{\socA(u-u_0)}_2\! =\! t$, where \vspace{-.5em}
\begin{equation}
\label{eq:least-square-u}
    u_0 = - \socAA^{-1} \socA^T \socb \vspace{-.3em}
\end{equation}
is the least-squares control input that minimizes $\norm{\socA u + \socb}_2$. We will refer to the conical surface $\beta\norm{\socA(u\!-\!u_0)}_2\! =\! t$ as the \emph{asymptote} of the hyperboloid. We now analyze each of the individual cases of Theorem \ref{th:nec_and_suf}.

\textit{Case 1 (Hyperbolic)}: This case matches the sufficient condition of Lemma \ref{lemma:sufficient}. Note that this condition implies that the necessary condition \eqref{eq:necessary} is trivially satisfied. In this case, the slope of the hyperplane $\socc u + \socd = t$ is greater than the slope of the asymptote of the hyperboloid for the direction of $u$ corresponding to $\usafe$.

\textit{Case 2 (Elliptic)}: Given that the smallest eigenvalue of $\feasmat$ is positive, then $\feasmat\succ 0$. Note that $\feasmat$ is the lower-right block of the matrix $H(x|\mathbb{D}_N)$ in \eqref{eq:H_def}. Therefore, $\feasmat\succ 0$ implies that the left-hand side of Equation \eqref{eq:feas_lemma1_1} must be strictly convex, with a unique global minimum at some $u = u_1 \in \R^m$. The first-order optimality condition gives \vspace{-.5em}
\begin{align}
\vspace{-.5em}
     u_1 = &-\feasmat^{-1}h,
     \label{eq:theorem2u1}
\end{align}
with $h\!:=\!\beta^2\socA^T\socb\!-\!\socc^T$ ~$\big(\socd\big)$.
 Since at $u_1$ the minimum is attained, Equation \eqref{eq:feas_lemma1_1} holds if and only if \vspace{-.3em}
\begin{equation}
\label{eq:condition_u1}
    [1\ u_1^T] H(x|\mathbb{D}_N) \begin{bmatrix} 1 \\u_1 \end{bmatrix} \leq 0. \vspace{-.3em}
\end{equation}
Plugging \eqref{eq:H_def} and \eqref{eq:theorem2u1} into \eqref{eq:condition_u1}, we get \vspace{-.3em}
\begin{multline}
\label{eq:condition_h}
    \beta^2\socbb - \big(\socd\big)^2 - \\ h^T \feasmat^{-1} h = H(x|\mathbb{D}_N) / \feasmat \leq 0.\vspace{-.5em}
\end{multline}

\noindent Since for this case $\feasmat$ is positive definite, and $H(x|\mathbb{D}_N)$ cannot be positive definite by the necessary condition \eqref{eq:necessary}, then from \cite[Thm. 1.12]{zhang2006schur} the inequality \eqref{eq:condition_h} must be satisfied. 
Consequently, \eqref{eq:feas_lemma1_1} holds for $u=u_1$.
Now, plugging $u_1$ into \eqref{eq:feas_lemma1_2} we have: $\socc u_1 + \socd$ equals the left-hand side of \eqref{eq:case2-add}. Thus, from Lemma \ref{lemma:fundamental} the feasible set is non-empty if and only if \eqref{eq:case2-add} is non-negative. On the other hand, \eqref{eq:case2-add} being negative would mean that the hyperplane $\socc u \! + \! \socd \! =\! t$ intersects the hyperboloid's negative sheet, $-\beta\norm{\socA u \! + \! \socb}_2\!=\!t$, forming an ellipse, and therefore cannot intersect the positive sheet.
Consequently, when $\eigval(x|\mathbb{D}_N) > 0$, the GP-CBF-SOCP \eqref{eq:gp-cbf-socp} is feasible if and only if \eqref{eq:case2-add} holds.

\textit{Case 3 (Parabolic)}: For this case, $\eigval(x|\mathbb{D}_N)\!=\!0$ means that there exists some control input direction for which the hyperplane and the asymptote have the same slope. Define \vspace{-.5em}
\begin{multline*}
    p \coloneqq \socd - \\ 
    \socc\socAA^{-1}\socA^T \socb. \vspace{-.3em}
\end{multline*}
Then, condition \eqref{eq:case3-add} is satisfied if and only if $\!p\!>\!0$. Consider the control input $u=u_0$ from \eqref{eq:least-square-u} that minimizes $\norm{\socA u + \socb}_2$. Then, we can rewrite $p = \socd + \socc u_0 $.

Furthermore, let $\eigvec(x|\mathbb{D}_N)$ denote the unit eigenvector of $\feasmat$ associated with the eigenvalue $\eigval(x|\mathbb{D}_N)=0$. Then, clearly, $\eigvec(x|\mathbb{D}_N)^T \feasmat \eigvec(x|\mathbb{D}_N) = 0$. Based on the definition of $\feasmat$ \eqref{eq:defineF}, since $\socAA \succ 0$ then it must hold that $\socc \eigvec(x|\mathbb{D}_N) \neq 0$. Next, using a control input of the form \vspace{-.5em}
\begin{equation}
\label{eq:soc-case3-u}
    u = u_0 + \alpha \text{sgn}(\socc \eigvec(x|\mathbb{D}_N)) \eigvec(x|\mathbb{D}_N), \quad \alpha > 0, \vspace{-.5em}
\end{equation}
we can write the left-hand side of \eqref{eq:feas_lemma1_1}, as \vspace{-0.5em}
\begin{align}
\label{eq:soc-proof-sub1}
\begin{split}
    &\hspace{-7pt} \beta^2\socbb-\!(\socd)^2\!+\!2 h^T u_0+ \\ & u_{0}^T \feasmat u_0 - 2 \alpha p |\socc \eigvec(x|\mathbb{D}_N)|.
\end{split}
\end{align}
And plugging \eqref{eq:soc-case3-u} into the left-hand side of \eqref{eq:feas_lemma1_2}, we obtain\vspace{-.3em}
\begin{align}
\label{eq:soc-proof-sub2}
    p + \alpha \cdot |\socc \eigvec|. 
    \vspace{-.8em}
\end{align}
If $p$ is positive, then there exists a large-enough positive constant $\alpha$ such that \eqref{eq:soc-proof-sub1} is non-positive and \eqref{eq:soc-proof-sub2} positive. Therefore, from Lemma~\ref{lemma:fundamental}, the GP-CBF-SOCP \eqref{eq:gp-cbf-socp} is feasible. 

Note that the geometric interpretation of the condition $p\!>\!0$ is that the hyperplane $\socc u\! +\! \socd\! =\! t$,
which has the same slope as the asymptote of $\beta\norm{\socA u\! +\! \socb}_2\! =\! t$
along the direction of $\eigvec(x|\mathbb{D}_N)$, should be placed over the asymptote in order for it to intersect the positive sheet of $\beta\norm{\socA u\! +\! \socb}_2\! =\! t$.
Furthermore, at $u=u_0$, the asymptote $\beta\norm{\socA(u-u_0)}_2\! =\! t$ takes value $t\!=\!0$, and $p$ is the value of the hyperplane $\socc u\! +\! \socd\! =\! t$ at $u\! =\! u_0$. Therefore, when $p \le 0$, the hyperplane is always under the positive sheet of the hyperboloid, and never intersects it. Consequently, the constraint \eqref{eq:constraint_std} is not feasible when $p \le 0$.
\qed

\subsection{Proof of Lemma \ref{lemma:algo_lambda}}
\label{subsec:proof-lemma-lambda}
Let us consider the trajectory generated by running Algorithm \ref{algo:safelearning} from any $x_0 \in \mathcal{X}$, which we assume locally exists and is unique (as stated in the hypothesis of the Lemma). For a fixed dataset $\mathbb{D}_{N}$, $\eigval\big(x|\mathbb{D}_{N})$ is a continuous function of the state $x$ by basic continuity arguments.
For the event-triggered updates of the dataset, 
if at time $t$ we have $\eigval\big(x(t)|\mathbb{D}_{N(t)}) \geq -\thres$, then Algorithm \ref{algo:safelearning} applies a control input $\usafe(x(t))$ from \eqref{eq:safe-direction},
collects the resulting measurement, and adds it to $\mathbb{D}_{N(t)}$, forming $\mathbb{D}_{N(t)+1}$. Note that from the posterior variance expression \eqref{eq:sigma_adp}, after adding the new data point (with $\usafe(x(t))$ as input) we have 
\begin{equation}
\label{eq:lemma5proof1}
\eigvec(x|\mathbb{D}_{N(t)})^T \gpGramUB(x|\mathbb{D}_{N(t)+1}) \eigvec(x|\mathbb{D}_{N(t)}) \to 0
\end{equation}
for large $\alpha$ in $\usafe$. Furthermore, we know that 
\begin{equation}
\footnotesize
\label{eq:lemma5proof2}
\eigvec(x|\mathbb{D}_{N(t)})^T \big( \wfeas(x|\mathbb{D}_{N(t)+1}) \wfeas(x|\mathbb{D}_{N(t)+1})^T\big) \eigvec(x|\mathbb{D}_{N(t)}) > 0,
\end{equation}
since from Assumptions \ref{assumption:reldegree1} and \ref{assumption:triggering_edagger}, $\wfeas(x(t)|\mathbb{D}_{N(t) + 1}) \neq 0$ and $\wfeas\big(x(t)|\mathbb{D}_{N(t)+1}\big) \eigvec\big(x(t)|\mathbb{D}_{N(t)}\big) \neq 0$. The statements \eqref{eq:lemma5proof1} and \eqref{eq:lemma5proof2} mean that we can find a large-enough constant $\alpha>0$ such that $\eigvec(x|\mathbb{D}_{N(t)})^T \big(\beta^2 \gpGramUB(x|\mathbb{D}_{N(t)+1}) - \wfeas(x|\mathbb{D}_{N(t)+1}) \wfeas(x|\mathbb{D}_{N(t)+1})^T\big) \eigvec(x|\mathbb{D}_{N(t)}) <0$, leading to $\eigval\big(x(t)|\mathbb{D}_{N(t)+1})<0$. An equivalent argument proves that with the time-triggered updates $\eigval$ stays negative after the new data point is added.
\qed

\subsection{Proof of Lemma \ref{lemma:lipsch_SOCP}}
\label{subsec:proof-lemma-lipschitz}
We use \cite[Thm. 6.4]{still2018lectures} which provides a sufficient condition for local Lipschitz continuity of solutions of parametric optimization problems. Twice differentiability of the objective and constraints with respect to both state and input trivially follows from Assumption \ref{assumption:contdiff} and the structure of \eqref{eq:gp-cbf-socp}. For a given state $x \in \mathcal{X}$ and dataset $\mathbb{D}_N$, $\eigval(x|\mathbb{D}_N)<0$ means that there exists a control input $\usafe$ from \eqref{eq:safe-direction} that strictly satisfies constraint \eqref{eq:socp-cbf-constraint}, meaning that in this case \eqref{eq:gp-cbf-socp} satisfies Slater's Condition (SC), since the problem is convex. \cite[Thm. 6.4]{still2018lectures} requires satisfaction of the Mangasarian Fromovitz Constraint Qualification (MFCQ) and the Second Order Condition (SOC2) of \cite[Def. 6.1]{still2018lectures} at the solution of $\eqref{eq:gp-cbf-socp}$. In \cite[Prop. 5.39]{andreasson2020introduction}, it is shown that SC implies MFCQ. Furthermore, since we have a strongly convex objective function in the decision variables $(u,d)$, and the constraints are convex in $(u,d)$, the Lagrangian of \eqref{eq:gp-cbf-socp} is strongly convex in $(u,d)$, implying SOC2 satisfaction. \qed

\subsection{Proof of Theorem \ref{thm:existence_uniqueness}}
\label{subsec:proof-theorem-existence}
The proof follows from \cite[Thm. III.1]{lygeros2003dynamical} using local Lipschitz continuity of the continuous dynamics (from Lemma \ref{lemma:lipsch_SOCP} and the expression of $\usafe$ in \eqref{eq:safe-direction}) instead of global Lipschitz continuity, therefore establishing local existence and uniqueness of executions of the closed-loop switched system. \qed

\end{appendix}

\bibliographystyle{IEEEtran}
\bibliography{reference.bib}

\begin{IEEEbiography}[{\includegraphics[width=1in,height=1.25in,clip,keepaspectratio]{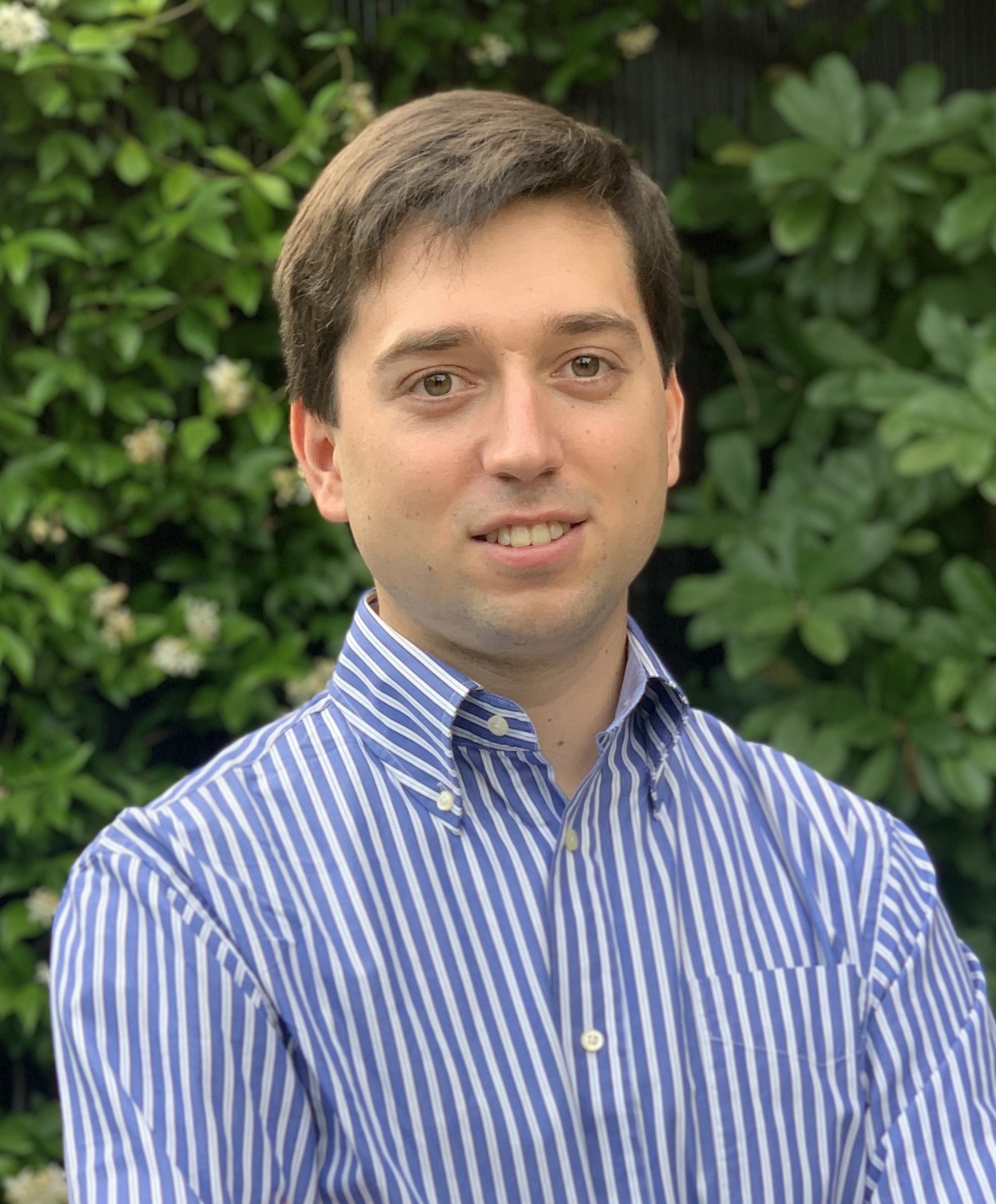}}]{Fernando Castañeda} received a Ph.D. degree in mechanical engineering from the University of California, Berkeley in 2023. He is a Rafael del Pino Foundation Fellow (2020) and a ``la Caixa" Foundation Fellow (2017). His research interests lie at the intersection of nonlinear control and data-driven methods, with a particular emphasis on ensuring the safe operation of high-dimensional systems in the real world. 
\end{IEEEbiography}

\begin{IEEEbiography}
[{\includegraphics[width=1in,height=1.25in,clip,keepaspectratio]{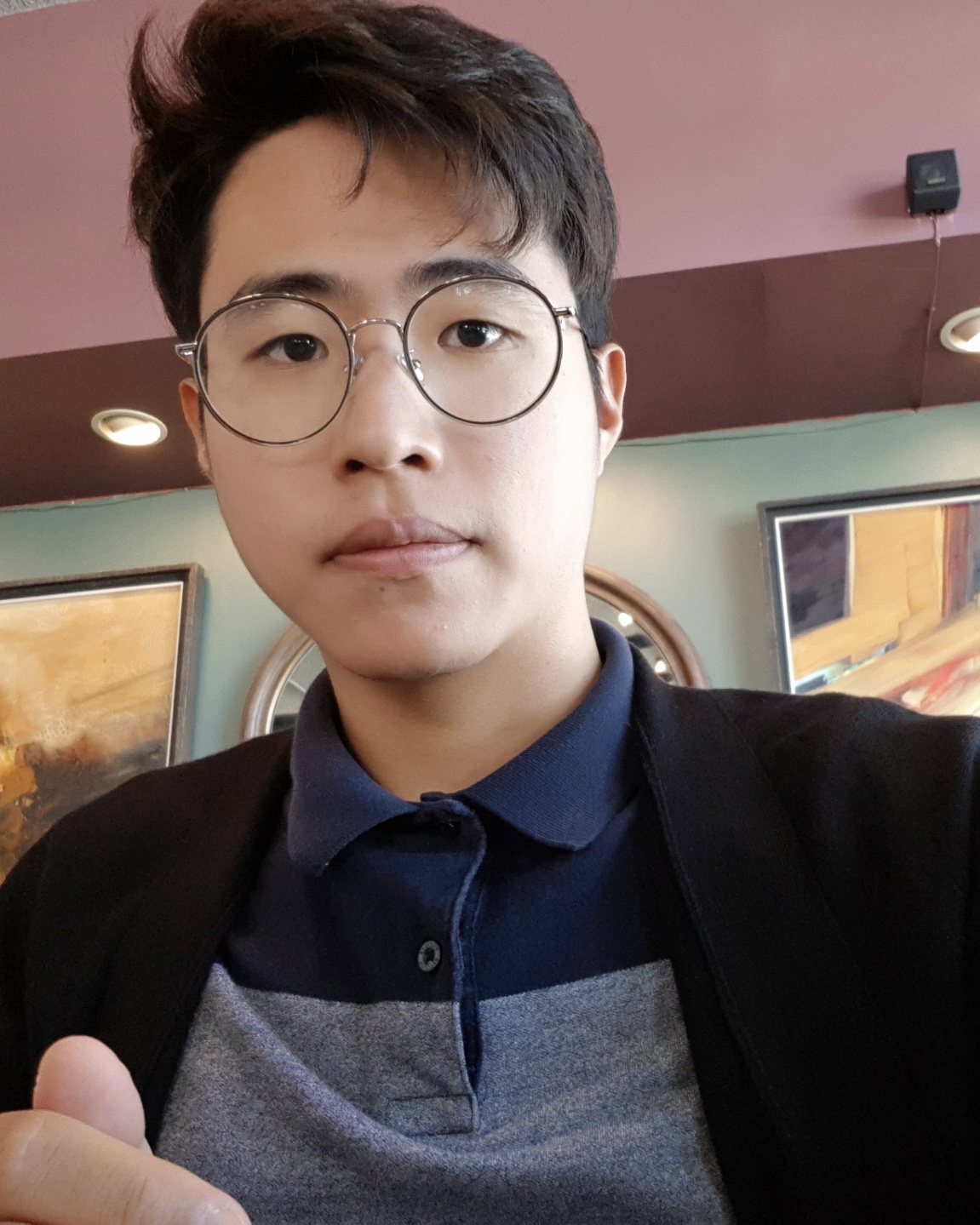}}]{Jason J. Choi} (Student Member, IEEE) received the B.S. degree in mechanical engineering from Seoul National University in 2019. He is currently pursuing a Ph.D. degree at University of California Berkeley in mechanical engineering. His research interests center on optimal control theories for nonlinear and hybrid systems, data-driven methods for safe control, and their applications to robotics and autonomous mobility.    
\end{IEEEbiography}

\begin{IEEEbiography}
[{\includegraphics[width=1in,height=1.25in,clip,keepaspectratio]{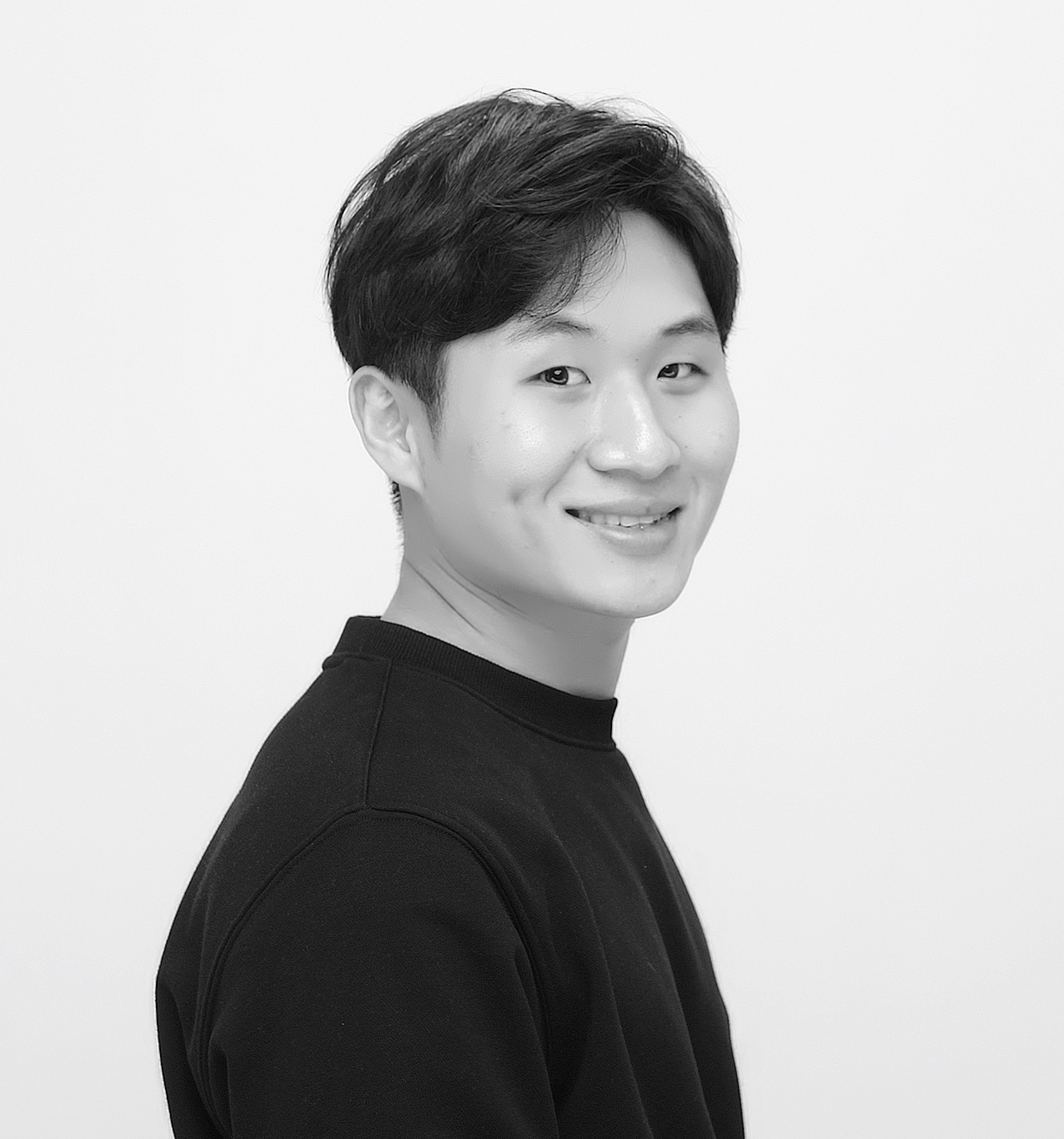}}]{Wonsuhk Jung} (Student Member, IEEE) received the B.S. degree in Mechanical Engineering and Artificial Intelligence from Seoul National University in 2022. He is currently pursuing a Ph.D. degree at Georgia Institute of Technology in Robotics. His research focuses on leveraging optimal control, planning, and data-driven methodologies for the safe operation of contact-rich robotics platforms.
\end{IEEEbiography}

\begin{IEEEbiography}
[{\includegraphics[width=1in,height=1.25in,clip,keepaspectratio]{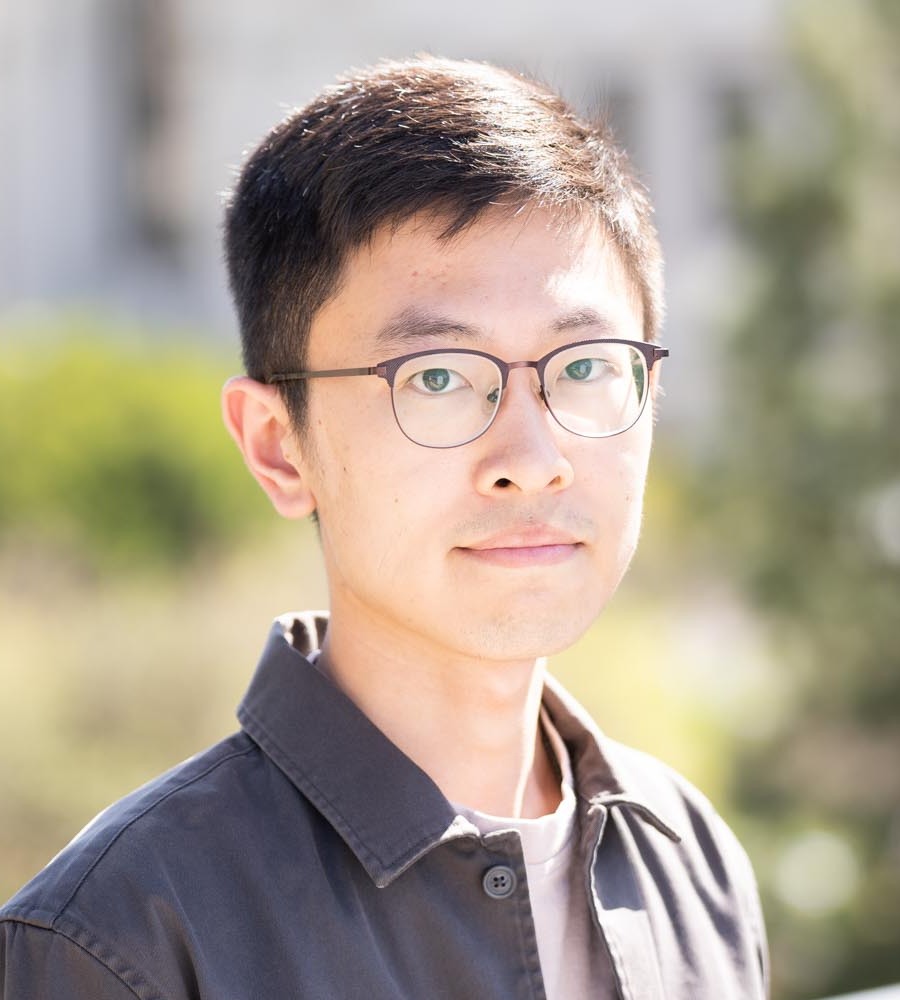}}]{Bike Zhang} (Student Member, IEEE) received the B.Eng. degree in electrical engineering and automation from Huazhong University of Science and Technology in 2017. He is currently working toward the Ph.D. degree in mechanical engineering at University of California Berkeley. His current research interests include predictive control and reinforcement learning with application to legged robotics.
\end{IEEEbiography}

\begin{IEEEbiography}
[{\includegraphics[width=1in,height=1.25in,clip,keepaspectratio]{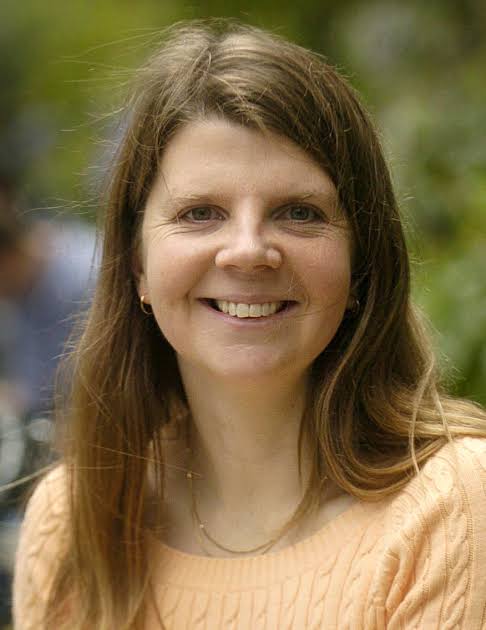}}]
{{C}laire J. Tomlin}{\,}(Fellow, IEEE) is the James and Katherine Lau Professor of Engineering and professor and chair of the Department of Electrical Engineering and Computer Sciences (EECS) at UC Berkeley. She was an assistant, associate, and full professor in aeronautics and astronautics at Stanford University from 1998 to 2007, and in 2005, she joined UC Berkeley. She works in the area of control theory and hybrid systems, with applications to air traffic management, UAV systems, energy, robotics, and systems biology. She is a MacArthur Foundation Fellow (2006), an IEEE Fellow (2010), and in 2017, she was awarded the IEEE Transportation Technologies Award. In 2019, Claire was elected to the National Academy of Engineering and the American Academy of Arts and Sciences.
\end{IEEEbiography}

\begin{IEEEbiography}
[{\includegraphics[width=1in,height=1.25in,clip,keepaspectratio]{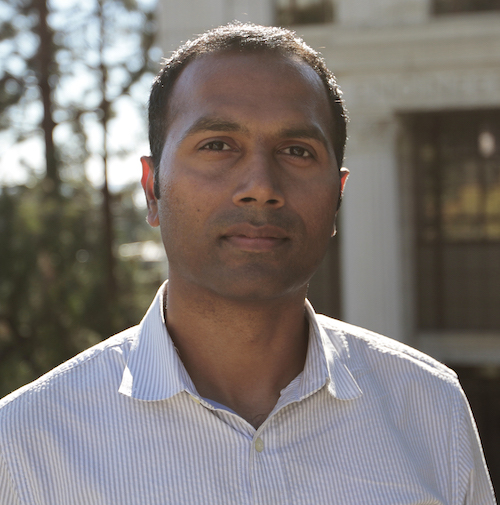}}]
{{K}oushil Sreenath}{\,}(Member, IEEE) is an associate professor of mechanical engineering, at UC Berkeley. He received a Ph.D. degree in electrical engineering and computer science and a M.S. degree in applied mathematics from the University of Michigan at Ann Arbor, MI, in 2011. He was a postdoctoral scholar at the GRASP Lab at University of Pennsylvania from 2011 to 2013 and an assistant professor at Carnegie Mellon University from 2013 to 2017. His research interest lies at the intersection of highly dynamic robotics and applied nonlinear control. He received the NSF CAREER, Hellman Fellow, Best Paper Award at the Robotics: Science and Systems (RSS), and the Google Faculty Research Award in Robotics.
\end{IEEEbiography}
\end{document}